\documentclass[12pt,reqno]{amsart}
\usepackage{amssymb}
\usepackage{amsfonts}
\usepackage{indentfirst}
\usepackage{amsopn}
\usepackage{amsmath}
\usepackage{amsthm}
\usepackage{amscd}
\usepackage[dvipdf]{graphicx}
\usepackage{color,xcolor}

\makeatletter\@addtoreset {equation}{section}\makeatother

\newtheorem{corollary}{Corollary}

\newtheorem{lemma}{Lemma}
\newtheorem{remark}{Remark}

\begin{document}

\title[Standing periodic waves in the Ablowitz-Ladik equation]{\bf Rogue waves arising on the standing periodic waves \\ in the Ablowitz-Ladik equation}

\author{Jinbing Chen}
\address[J. Chen]{School of Mathematics, Southeast University, Nanjing, Jiangsu 210096, P.R. China}
\email{cjb@seu.edu.cn}
	
\author{Dmitry E. Pelinovsky}
\address[D.E. Pelinovsky, Corresponding author]{
	Department of Mathematics, McMaster University, Hamilton, Ontario, Canada, L8S 4K1 }
\email{dmpeli@math.mcmaster.ca}


\begin{abstract}
	We study the standing periodic waves in the semi-discrete integrable system modelled by the Ablowitz--Ladik equation. We have related the stability spectrum to the Lax spectrum by separating the variables and by finding the characteristic polynomial for the standing periodic waves. We have also obtained rogue waves on the background of the modulationally unstable standing periodic waves by using the end points of spectral bands and the corresponding eigenfunctions. The magnification factors for the rogue waves have been computed analytically and compared with their continuous counterparts. The main novelty of this work is that we explore a non-standard linear Lax system, which is different from the standard Lax representation of the Ablowitz--Ladik equation.
\end{abstract}

\date{\today}
\maketitle

{\em This article is dedicated to Athanassios S. Fokas on the occasion of his 70th anniversary for his many contributions to studies of integrable nonlinear PDEs and boundary-value problems.}

\section{Introduction}
\label{sec:1}

The nonlinear Schr\"{o}dinger (NLS) equation models wave dynamics in many physical problems related to fluids, plasmas, and optics \cite{Fibich-book,Sulem-book}. Complicated wave patterns can be expressed analytically by using exact solutions of the NLS equation for periodic and double-periodic standing waves (see review in \cite{Pel-Front}). The standing periodic waves are known to be modulationally unstable \cite{DS,DU} and rogue waves (localized perturbations in space and time) have been observed on their backgrounds in numerical experiments \cite{AZ,Kedrosa}. Rogue waves are generated due to modulational instability of the wave background \cite{Suret,Amin}. The exact solutions for rogue waves arising on the periodic 
stading waves were obtained analytically \cite{CPnls,CPW,Feng} 
and confirmed experimentally for fluids and optics \cite{XuKibler}.

{\em It is natural to ask if the modulational instability and rogue waves persist on the standing periodic waves in the integrable discretizations of the integrable NLS equation.} This question has received much less attention 
in the literature. The main purpose of our work is to answer this question  
for the Ablowitz-Ladik (AL) equation \cite{AL-1976} written 
in the normalized form:
\begin{equation}
i \dot{u}_n + (1 + |u_n|^2) (u_{n+1} + u_{n-1}) = 0,
\label{al}
\end{equation}
where the dot represents the derivative of $\{ u_n(t) \}_{n \in \mathbb{Z}} \in \mathbb{C}^{\mathbb{Z}}$
with respect to the time variable $t \in \mathbb{R}$, and $i = \sqrt{-1}$.
The continuum limit of the AL equation (\ref{al}) is obtained for slowly varying wave packets of small amplitude, the leading order of which can be represented as
\begin{equation}
\label{exp-continuum}
u_n(t) = \varepsilon \mathfrak{u}(\varepsilon n, \varepsilon^2 t) e^{2it},
\end{equation}
where $\varepsilon > 0$ is a formal small parameter. Substituting (\ref{exp-continuum}) into (\ref{al}) and expanding
$\mathfrak{u}(X\pm \varepsilon,T)$ with $X = \varepsilon n$ and $T = \varepsilon^2 t$ in the Taylor series in $\varepsilon$ yield
at the formal order of $\mathcal{O}(\varepsilon^3)$ the continuous NLS equation 
in the form:
\begin{equation}
i \mathfrak{u}_T + \mathfrak{u}_{XX} + 2 |\mathfrak{u}|^2 \mathfrak{u} = 0.
\label{nls}
\end{equation}

Rogue waves on the constant-amplitude wave background have been obtained 
for the AL equation (\ref{al}) and related discrete equations in \cite{AkhAnk2,AkhAnk1,YangOhta,ZhaoYu}. 
They have been observed in numerical experiments \cite{Agafontsev}. Higher-order rogue waves have been studied by using the inverse scattering 
method \cite{Baofeng}. Further generalization of these rogue waves to fully discrete integrable NLS equation was recently given in \cite{OhtaFeng}. {\em What we will develop in this work is the construction of rogue waves on the background 
of discrete standing periodic waves.} The discrete periodic and double-periodic waves were obtained previously for the AL equation in \cite{Chow,Huang,8}.

In a series of recent works, we have constructed rogue waves on the background of standing periodic waves in the continuous integrable models including the NLS equation \cite{CPnls,CPW}, the derivative NLS equation \cite{Chen-DNLS,2}, the modified KdV equation \cite{CPkdv,Chen-JNLS}, and the sine--Gordon equation \cite{PW}. The latest work in this series was the first construction of such  rogue waves in the discrete modified KdV equation \cite{Chen-Pel-2022}. In all works, we used the nonlinearization method \cite{5,6} which allowed us to characterize the standing periodic waves as the 
restriction of solutions of nonlinear models as squared solutions 
of the linear Lax equations. 

However, we have failed to characterize the discrete standing periodic waves 
by using constraints of the nonlinearization method for the AL equation 
(developed in \cite{3,7}) because the constraints generate difference equations which are not satisfied by the standing periodic waves. As a result, we had to develop new ideas based on separation of variables for the standing periodic waves 
in the linear Lax equations. This separation of variables is similar 
to the approach in \cite{DS} used to characterize the modulational stability 
of standing periodic waves in the continuous NLS equation. With this approach, 
we can obtain analytically the end points of spectral bands of the Lax spectrum and the corresponding eigenfunctions, which are then used to obtain 
the rogue wave solutions on the background of discrete standing periodic waves. 

Integrability of the AL equation can be expressed by using 
two different Lax formulations. The standard formulation explored previously in the context of rogue waves on the constant-amplitude wave background \cite{AkhAnk2,Baofeng,YangOhta} is irrelevant for the modulational stability and rogue waves on the standing periodic waves due to several issues: (i) the location of the Lax spectrum is different, (ii) there exists no relation between squared eigenfunctions and the linearized AL equation, and consequently, (iii) the rogue waves are not properly defined. On the other hand, we show here that the alternative formulation related to symplectic 
integrable map of the nonlinearization method in \cite{3} allows us to fix issues (i), (ii), and (iii) 
and to describe properties of the rogue waves on the 
standing periodic wave background. {\em The main novelty of our work  is that our results are derived from non-standard Lax formulation of the AL equation. }

Among possible applications of our main results, we can mention recent studies 
of breathers and rogue waves in non-integrable discrete settings \cite{Sul}. 
Many numerical results are based on the homotopy continuation of breathers 
and rogue waves from their integrable limits expressed by the AL equation.
With the precise construction of the standing periodic waves and their rogue waves, we can then utilize these solutions in the homotopy continuation towards 
the nonintegrable versions of the discrete NLS equations. 

The article is organized as follows. Section \ref{sec:2} presents the main results of our computations. Section \ref{sec:3} establishes the relation between squared eigenfunctions of the linear Lax equations and solutions of the linearized AL equation. Section \ref{sec:4} presents analytical and numerical approximations of the Lax and stability spectra. Section \ref{sec:6} gives construction of nonperiodic solutions of the linear Lax equations in terms of periodic eigenfunctions. Section \ref{sec:5} explains how the Darboux transformation is used to 
obtain the rogue wave solutions from the nonperiodic solutions of the linear Lax equations and to study magnification factors of the rogue waves.

\section{Lax pair, standing periodic waves, and rogue waves}
\label{sec:2}

The AL equation (\ref{al}) can be represented as the compatibility condition for a Lax pair of linear equations:
\begin{equation}\label{lax-intro}
\varphi_{n+1} = U(u_n,\lambda) \varphi_n, \qquad \dot{\varphi}_n = V(u_n,\lambda) \varphi_n,
\end{equation}
defined by the matrices
\begin{align*}
U(u_n,\lambda) = \frac{1}{\sqrt{1+ |u_n|^2}} \left(\begin{array}{cc}
\lambda & u_n\\
- \bar{u}_n & \lambda^{-1}\\
\end{array}
\right)
\end{align*}
and
\begin{align*}
 \qquad V(u_n,\lambda) = i \left(\begin{array}{cc}
 \frac12 \left(\lambda^2 + \lambda^{-2} + u_n \bar{u}_{n-1} + \bar{u}_n u_{n-1} \right) & \lambda u_n - \lambda^{-1} u_{n-1} \\
- \lambda \bar{u}_{n-1} + \lambda^{-1} \bar{u}_n &  -\frac12 \left(\lambda^2 + \lambda^{-2} + u_n \bar{u}_{n-1} + \bar{u}_n u_{n-1} \right)\\
\end{array}
\right),
\end{align*}
where $\{ \varphi_n\}_{n \in \mathbb{Z}} \in (\mathbb{C}^2)^{\mathbb{Z}}$ depends on time $t$, $\lambda \in \mathbb{C}$ is a spectral parameter which is constant in $n$ and $t$, and $\bar{u}_n(t)$ denotes the complex-conjugate of $u_n(t)$.

\begin{remark}
There exists a simpler Lax pair, 	
\begin{equation}\label{lax-1-intro}
\varphi_{n+1} = U(u_n,\lambda) \varphi_n, \quad U(u_n,\lambda) = \left(\begin{array}{cc}
\lambda & u_n\\
- \bar{u}_n & \lambda^{-1}\\
\end{array}
\right),
\end{equation}
and
\begin{equation}\label{lax-2-intro}
\dot{\varphi}_n = V(u_n,\lambda) \varphi_n, \ V(u_n,\lambda) = i \left(\begin{array}{cc}
\frac12 \left(\lambda^2 + \lambda^{-2} \right) + u_n \bar{u}_{n-1} & \lambda u_n - \lambda^{-1} u_{n-1}\\
- \lambda \bar{u}_{n-1} + \lambda^{-1} \bar{u}_n & - \frac12 \left(\lambda^2 + \lambda^{-2} \right) - \bar{u}_n u_{n-1}\\
\end{array}
\right),
\end{equation}
commutativity of which also yields the AL equation (\ref{al}). We have found that the stability spectrum of the standing periodic waves can be characterized from the Lax spectrum associated with the linear system (\ref{lax-intro}) due to the squared eigenfunction relation, however, we did not find the squared eigenfunction relation for eigenfunctions of the linear system (\ref{lax-1-intro})--(\ref{lax-2-intro}).
\label{rem-lax}
\end{remark}

We consider the standing wave solution of the AL equation (\ref{al}) in the form:
\begin{equation}
u_n(t) = U_n e^{2i \omega t}, \label{2.1}
\end{equation}
where $\{ U_n \}_{n \in \mathbb{Z}} \in \mathbb{C}^{\mathbb{Z}}$ and $\omega \in \mathbb{R}$ is a frequency parameter. The AL equation \eqref{al} with the standing wave reduction (\ref{2.1}) becomes the second-order difference equation
\begin{equation} \label{2.2}
(1 + |U_n|^2) (U_{n+1} + U_{n-1}) = 2 \omega U_n, \quad n \in \mathbb{Z}.
\end{equation}
It can be integrated with two conserved quantities as in the following lemma.

\begin{lemma}
	\label{lem-conserved}
	Let $\{ U_n \}_{n \in \mathbb{Z}} \in \mathbb{C}^{\mathbb{Z}}$  be a solution of the difference equation (\ref{2.2}). Then, the following real-valued quantities
	\begin{equation} \label{2.2-1}
	F_0 := i (U_n \bar{U}_{n-1} - \bar{U}_n U_{n-1})
	\end{equation}
	and
	\begin{equation} \label{2.2-2}
	F_1 := \omega (U_n \bar{U}_{n-1} + \bar{U}_n U_{n-1}) - |U_n|^2 -  |U_{n-1}|^2 - |U_n|^2 |U_{n-1}|^2
	\end{equation}
	are constant in $n \in \mathbb{Z}$.
\end{lemma}

\begin{proof}
By multiplying (\ref{2.2}) by $\bar{U}_n$ and subtracting the complex conjugate, we obtain
$$
(1+|U_n|^2) [ (U_{n+1} \bar{U}_n - \bar{U}_{n+1} U_n) - (U_{n} \bar{U}_{n-1} - \bar{U}_{n} U_{n-1})] = 0,
$$
from which conservation of $F_0$ follows.

Similarly, multiplying (\ref{2.2}) by $\bar{U}_{n-1}$ and by $\bar{U}_{n+1}$ and adding the complex conjugate yield
$$
(1+|U_n|^2) |U_{n+1}|^2 - \omega (U_n \bar{U}_{n+1} + \bar{U}_{n} U_{n+1}) =
(1+|U_n|^2) |U_{n-1}|^2 - \omega (U_n \bar{U}_{n-1} + \bar{U}_{n} U_{n-1}),
$$
from which conservation of $F_1$ follows.
\end{proof}

Periodic waves with trivial phase are given by real-valued periodic solutions of the difference equation (\ref{2.2}) with $F_0 = 0$ in (\ref{2.2-1}). Two families were obtained in \cite{Chen-Pel-2022} (see also \cite{Huang,8}), which we refer to as the {\bf dnoidal} and {\bf cnoidal} waves.
The dnoidal waves are given in the form:
\begin{equation}
\label{2.3}
U_n =  \frac{{\rm sn}(\alpha; k)}{{\rm cn}(\alpha; k)} {\rm dn}(\alpha n;k), \qquad
\omega = \frac{{\rm dn}(\alpha;k)}{{\rm cn}^2(\alpha;k)},
\end{equation}
with two parameters of $\alpha \in (0,K(k))$ and $k \in (0,1)$, where $k$ is the elliptic modulus.
\begin{itemize}
	\item As $k \to 0$, the dnoidal wave degenerates to the constant-amplitude wave
	\begin{equation}
	\label{2.4}
	U_n = \tan(\alpha), \qquad \omega = {\rm sec}^2(\alpha), \qquad \alpha \in \left(0,\frac{\pi}{2}\right),
	\end{equation}
	\item As $k \to 1$,  the dnoidal wave degenerates into the solitary wave
	\begin{equation}
	\label{2.5}
	U_n = \sinh(\alpha) {\rm sech}(\alpha n), \qquad
	\omega = \cosh(\alpha),  \qquad \alpha \in (0,\infty).
	\end{equation}
\end{itemize}
The cnoidal waves are given in the form:
\begin{equation}\label{2.6}
U_n = \frac{k {\rm sn}(\alpha;k)}{{\rm dn}(\alpha;k)} {\rm cn}(\alpha n;k), \qquad
\omega = \frac{{\rm cn}(\alpha;k)}{ {\rm dn}^2(\alpha;k)},
\end{equation}
with two parameters of $\alpha \in (0,2K(k))$ and $k \in (0,1)$. The limit $k \to 0$ gives the trivial solution, whereas the limit $k \to 1$ gives the same solitary wave (\ref{2.5}).

The standing waves in the limit $\alpha \to 0$ recover the asymptotic approximation (\ref{exp-continuum}) with $\varepsilon = \alpha$.
Since
\begin{align*}
{\rm sn}(\alpha;k) = \alpha + \mathcal{O}(\alpha^3), \quad
{\rm cn}(\alpha;k) = 1 - \frac{1}{2} \alpha^2 + \mathcal{O}(\alpha^4), \quad
{\rm dn}(\alpha;k) = 1 - \frac{1}{2} k^2 \alpha^2 + \mathcal{O}(\alpha^4),
\end{align*}
the dnoidal and cnoidal waves of the AL equation (\ref{al}) yield solutions
$$
\mathfrak{u}(X,T) = e^{i (2-k^2)T} {\rm dn}(X;k), \qquad
\mathfrak{u}(X,T) = e^{i (2k^2-1)T} k {\rm cn}(X;k),
$$
which are the dnoidal and cnoidal waves of the continuous NLS equation (\ref{nls}).

\begin{remark}
	Since the continuous NLS equation (\ref{nls}) also has a family of the standing periodic waves with nontrivial phase \cite{CPW,DS}, we conjecture that there exist complex-valued periodic solutions of the difference equation (\ref{2.2}) with $F_0 \neq 0$ in (\ref{2.2-1}). However, such solutions are not considered in our work.
\end{remark}

Figures \ref{fig-3} and \ref{fig-4} present the main results which are the rogue waves on the dnoidal and cnoidal wave 
backgrounds. We have obtained the rogue waves by using only analytical methods 
with the following steps:
\begin{itemize}
	\item[(i)] characterizing end points of the spectral bands associated with standing periodic waves, 
	\item[(ii)] computing the corresponding periodic eigenfunctions for the end-point eigenvalues,
	\item[(iii)] representing the second solutions of the linear system for the end-point eigenvalues,
	\item[(iv)] performing the 1-fold Darboux transformation (DT) on the standing periodic waves with the periodic and nonperiodic eigenfunctions in (ii) and (iii).
\end{itemize} 
Numerical approximations are only used to represent the rogue waves graphically.
\vspace{0.2cm}
 
\begin{figure}[htbp]
	\includegraphics[width=8cm,height=6cm]{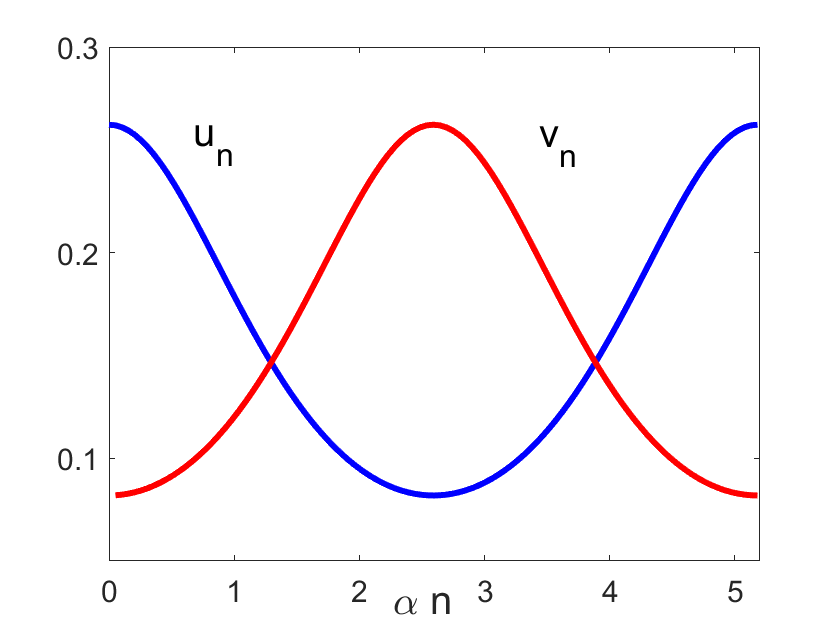}
	\includegraphics[width=8cm,height=6cm]{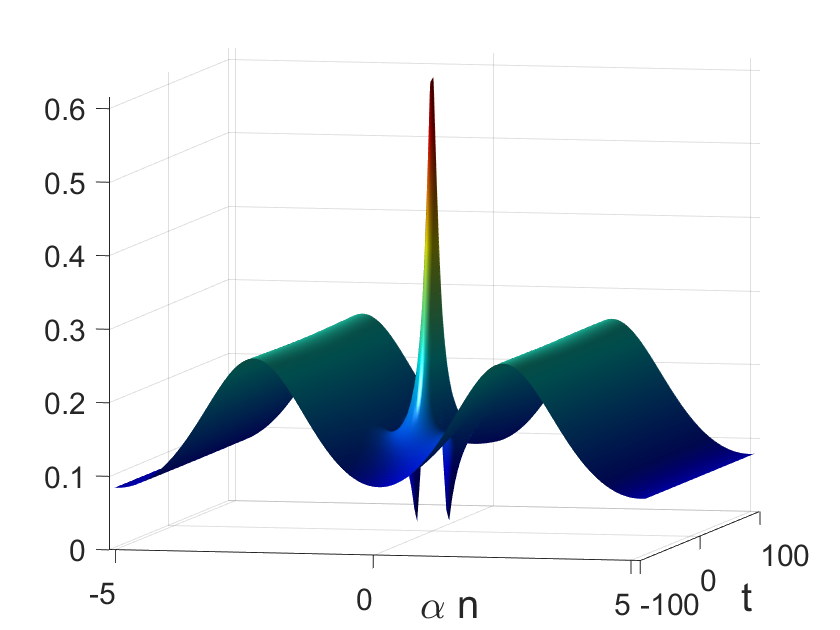}
	\includegraphics[width=8cm,height=6cm]{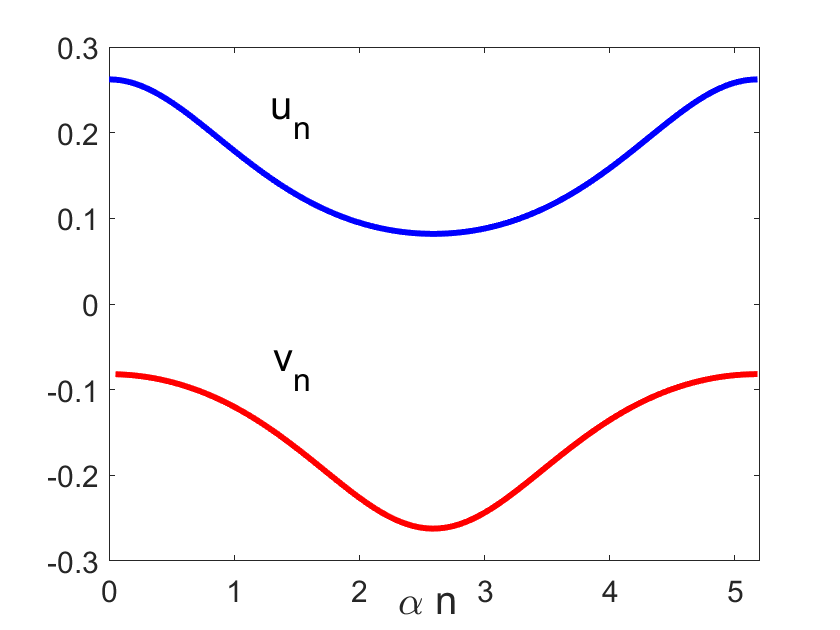}
	\includegraphics[width=8cm,height=6cm]{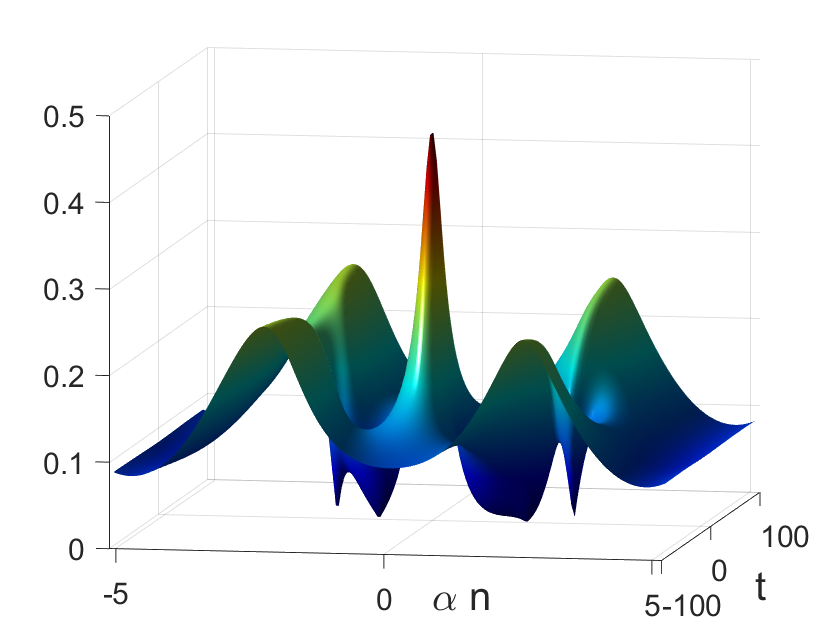}
	\caption{New solutions on the background of the dnoidal wave (\ref{2.3}). Left: profiles of $u_n$ and $v_n := \hat{u}_n$ versus $\alpha n$ obtained from the 1-fold DT with the periodic eigenfunctions. Right: solution surface $|\hat{u}_n(t)|$ versus $\alpha n$ and $t$ obtained from the 1-fold DT with the nonperiodic eigenfunctions. Top and bottom panels show rogue waves for two different end-point eigenvalues of the Lax spectrum. }
	\label{fig-3}
\end{figure}

Figure \ref{fig-3} shows the new solutions obtained by using the 1-fold DT on the background of the dnoidal wave (\ref{2.3}). The top panels correspond to the choice of $\lambda = \lambda_1$ and the bottom panels correspond to the choice of $\lambda = \lambda_2$, see (\ref{3-1-1}) and (\ref{3-1-2}) below. The left panels show the profiles of $u_n$ and $v_n$ versus $\alpha n$, where $u_n$ is the dnoidal wave (\ref{2.3}) and $v_n$ is a new solution obtained with the periodic eigenfunctions after the 1-fold DT. The new solution is a half-period translation of the original dnoidal wave with the sign flip for $\lambda = \lambda_2$ and with no sign flip for $\lambda = \lambda_1$.
The right panels show the solution surface of $|\hat{u}_n(t)|$ versus $\alpha n$ and $t$, where $\hat{u}_n$ is a new solution obtained with the nonperiodic eigenfunctions after the 1-fold DT.  The new solution is an isolated rogue wave on the background of the half-period translated dnoidal wave. 

As a practical outcome of the exact solutions, we can compute the magnification factor of the rogue waves as the quotient between the maximal amplitude of the rogue wave and the maximal amplitude of the dnoidal wave background. We have seen from  Figure \ref{fig-3}  that the maximal amplitude of the rogue wave is attained 
at $n = 0$ and $t = 0$, where we have computed analytically the magnification factor  in the closed form:
\begin{equation}
M_{\rm dn}(\alpha,k) =  1 + \frac{1 - \sigma_1 \sqrt
	{1-k^2} }{{\rm dn}(\alpha;k)},
\label{5.19}
\end{equation}
with $\sigma_1 = -1$ for $\lambda = \lambda_1$ and $\sigma_1 = +1$ for $\lambda = \lambda_2$. In the continuum limit $\alpha \to 0$, it converges to
$M_{\rm dn}(\alpha,k) \to 2 - \sigma_1 \sqrt{1-k^2}$, which reproduces
the correct expression for the magnification factor of the dnoidal wave
in the continuous NLS equation (\ref{nls}) obtained in \cite{CPnls}.

The rogue wave  associated with $\lambda_1$ has a bigger magnification factor than the one associated with $\lambda_2$. Since eigenvalues satisfy the order
$$
0 < \lambda_1 < \lambda_2 < 1 < \lambda_2^{-1} < \lambda_1^{-1},
$$
we associate this difference with the fact that $\lambda_1$ is located further from the unit circle compared to $\lambda_2$ in the Lax spectrum of the dnoidal wave, see Figure \ref{fig-1} below. The same rogue waves are computed from reflected eigenvalues $\lambda_1^{-1}$ and $\lambda_2^{-1}$ and from negative eigenvalues $-\lambda_1$, $-\lambda_2$, $-\lambda_1^{-1}$, and $-\lambda_2^{-1}$, which exist due to the symmetry of the Lax system (\ref{lax-intro}).

\begin{figure}[htbp]
	\includegraphics[width=8cm,height=6cm]{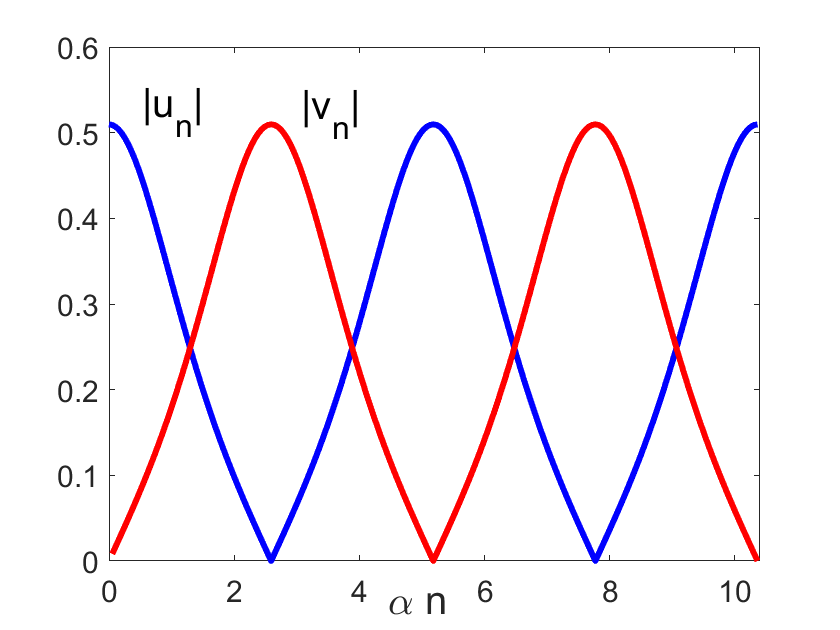}
	\includegraphics[width=8cm,height=6cm]{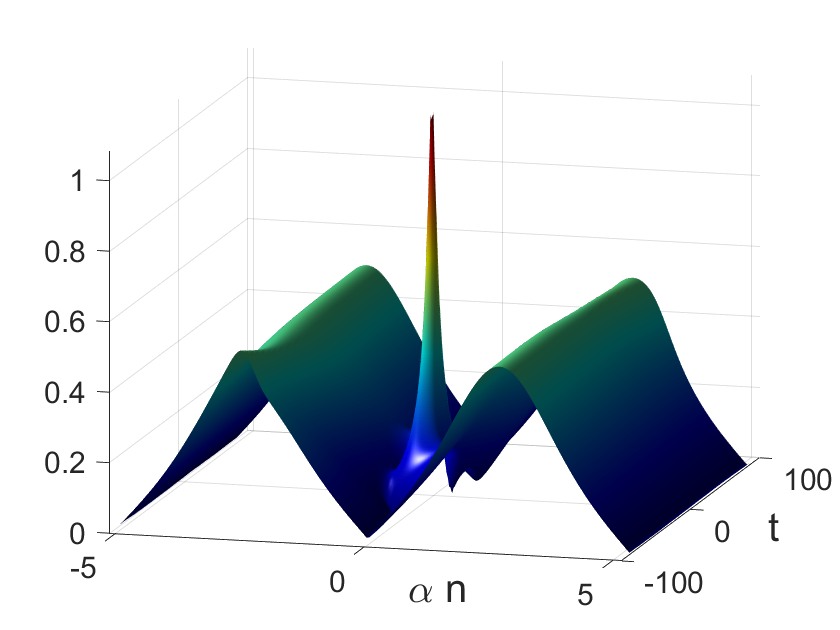}
	\caption{The same as on Figure \ref{fig-3} but for the cnoidal wave (\ref{2.6}) for one end-point eigenvalue of the Lax spectrum.}
	\label{fig-4}
\end{figure}

Figure \ref{fig-4} shows the new solutions obtained by using the 1-fold DT on the background of the cnoidal wave (\ref{2.6}). The profiles of $|u_n|$ and $|v_n|$ versus $\alpha n$ on a single period are shown on the left panel, where $v_n$ is a new solution obtained with the periodic eigenfunctions after the 1-fold DT. The new solution is a quarter-period translation of the original cnoidal wave with a phase factor. The solution surface of $|\hat{u}_n(t)|$ versus $\alpha n$ and $t$ is shown on the right panel, where $\hat{u}_n$ is a new solution obtained with the nonperiodic eigenfunctions after the 1-fold DT. The new solution is an isolated rogue wave on the background of the quarter-period translated cnoidal wave. Rogue waves associated with the reflected eigenvalues due to complex conjugation, reflection about the unit circle, and the sign reflection are displayed by similar solution surfaces since the complex phase is neglected in plotting of $|\hat{u}_n(t)|$.

We have observed from Figure \ref{fig-4} that the maximal
value of $|\hat{u}_n(t)|$ is not attained at $n = 0$ and $t = 0$. However, it is attained at a point located very close to the origin,
therefore, a good approximation of the magnification factor can still be computed at $n = 0$ and $t = 0$, where we have computed analytically the magnification factor for the rogue wave in the closed form:
\begin{align}
M_{\rm cn}(\alpha,k) = 1 + \frac{1}{ {\rm dn}(\alpha;k)}.
\label{5.23}
\end{align}
In the continuum limit $\alpha \to 0$, it converges to
$M_{\rm cn}(\alpha,k) \to 2$, which reproduces
the double magnification factor of the cnoidal wave
in the continuous NLS equation (\ref{nls}) obtained in \cite{CPnls}. 

Note that the magnification factors of the dnoidal and cnoidal waves for the continuous NLS equation (\ref{nls}) have been verified experimentally in \cite{XuKibler}. These main results suggest new experiments in the discrete setting modeled by the AL equation (\ref{al}) and other nonintegrable discretizations of the NLS equation (\ref{nls}).

Among other main results of this work, we mention the squared eigenfunction relation between solutions of the linearized AL equation and the linear Lax equations, see Lemma \ref{lem-spectrum} below. This allows us to connect the Lax spectrum and the stability spectrum and to characterize the periodic eigenfunctions of the linear Lax system from the standing periodic waves, 
see Corollary \ref{cor-squared-eigenfunction}, Corollary \ref{lem-stab}, 
and Lemma \ref{lem-relation} below. 
The Lax spectrum and the stability spectrum for the constant-amplitude wave 
are obtained analytically, see Lemma \ref{lem-constant-wave} below. 
However, we have to approximate the Lax spectrum numerically for the dnoidal and cnoidal waves. It follows from these numerical approximations that 
both the dnoidal and cnoidal waves are modulationally unstable. Finally, 
we give analytical expressions for the nonperiodic eigenfunctions in Lemma \ref{lem-second-solution} and \ref{lem-clear-expressions} below 
and we present the 1-fold DT in the closed form in Lemma \ref{lem-one-fold} below. 

\section{Spectral stability of  standing periodic waves}
\label{sec:3}

Here we will set up the spectral stability problem for the standing periodic waves and relate its eigenfunctions with squared eigenfunctions of the linear 
system (\ref{lax-intro}).

The spectral stability of the standing waves of the form \eqref{2.1} in the time evolution of the AL equation \eqref{al} can be studied by adding a perturbation of the form
\begin{equation}\label{3.1}
u_n(t) = e^{2i\omega t} [U_n + v_n(t)].
\end{equation}
Substituting \eqref{3.1} into the AL equation \eqref{al} and truncating at the linear terms in $v_n$ gives rise to the linearized AL equation
\begin{equation}\label{3.2}
i \dot{v}_n - 2 \omega v_n + (1 + |U_n|^2)(v_{n+1} + v_{n-1}) + (U_{n+1} + U_{n-1}) (U_n \bar{v}_n + \bar{U}_n v_n) = 0
\end{equation}
and its complex conjugate equation 
\begin{equation*}
-i \dot{\bar{v}}_n - 2 \omega \bar{v}_n + (1 + |U_n|^2)(\bar{v}_{n+1} + \bar{v}_{n-1}) + (\bar{U}_{n+1} + \bar{U}_{n-1}) (\bar{U}_n v_n + U_n \bar{v}_n) = 0.
\end{equation*}
Separation of variables with $v_n(t) = V_n e^{\Lambda t}$
and $\bar{v}_n(t) = \tilde{V}_n e^{\Lambda t}$, where $\tilde{V}_n$ is no longer a complex conjugate of $V_n$ if $\Lambda \notin \mathbb{R}$,
gives the spectral stability problem in the form:
\begin{equation}\label{spec-stab}
\left\{ \begin{array}{l} i \Lambda V_n - 2 \omega V_n + (1 + |U_n|^2)(V_{n+1} + V_{n-1}) + (U_{n+1} + U_{n-1}) (U_n \tilde{V}_n + \bar{U}_n V_n) = 0, \\
-i \Lambda \tilde{V}_n - 2 \omega \tilde{V}_n + (1 + |U_n|^2)(\tilde{V}_{n+1} + \tilde{V}_{n-1}) + (\bar{U}_{n+1} + \bar{U}_{n-1}) (\bar{U}_n V_n + U_n \tilde{V}_n) = 0.
\end{array} \right.
\end{equation}
Solutions of the spectral stability problem (\ref{spec-stab}) are computed by using the squared eigenfunctions of the linear system (\ref{lax-intro}).

Let us separate the variables for solutions $\varphi_n = (p_n,q_n)^\mathrm{T}$ of the linear system (\ref{lax-intro}) by using
\begin{equation}
p_n(t) = P_n(t) e^{i \omega t},\qquad q_n(t) = Q_n(t) e^{-i \omega t}. \label{3.3}
\end{equation}
Substituting \eqref{2.1} and \eqref{3.3} into \eqref{lax-intro} yields the following linear system:
\begin{equation}\label{3.4}
\left(\begin{array}{c}
P_{n+1}\\
Q_{n+1}\\
\end{array} \right) =  \frac{1}{\sqrt{1+ |U_n|^2}} \left(\begin{array}{cc}
\lambda & U_n\\
- \bar{U}_n & \lambda^{-1}
\end{array}
\right) \left(\begin{array}{c}
P_{n}\\
Q_{n}\\
\end{array} \right),
\end{equation}
and
\begin{equation}\label{3.5}
\frac{d}{dt} \left(\begin{array}{c}
{P}_{n}\\
{Q}_{n}\\
\end{array} \right)= i \left(\begin{array}{cc}
W_n - \omega & \lambda U_n - \lambda^{-1} U_{n-1}\\
- \lambda \bar{U}_{n-1} + \lambda^{-1} \bar{U}_n & \omega - W_n \\
\end{array}
\right) \left(\begin{array}{c}
P_{n}\\
Q_{n}\\
\end{array} \right),
\end{equation}
where $W_n := \frac12 \left(\lambda^2 + \lambda^{-2} + U_n \bar{U}_{n-1} + \bar{U}_n U_{n-1} \right)$.

The following lemma presents the squared eigenfunction relation between
solutions of the linearized AL equation (\ref{3.2})
and the squared eigenfunctions of the linear system (\ref{3.4})--(\ref{3.5}).

\begin{lemma}
Let $\{(P_n(t),Q_n(t))^{\mathrm{T}} \}_{n \in \mathbb{Z}}$ be a classical solution of the linear system  (\ref{3.4})--(\ref{3.5}) with an arbitrary $\lambda \in \mathbb{C}$ and a solution $\{ U_n \}_{n \in \mathbb{Z}}$ of the difference equation (\ref{2.2}). Then, $\{ v_n(t)\}_{n \in \mathbb{Z}}$ given by 
\begin{equation}
\label{squared-eigen}
v_n = \lambda P_n^2 - \bar{\lambda}^{-1} \bar{Q}_n^2 + U_n (P_n Q_n + \bar{P}_n \bar{Q}_n), \quad n \in \mathbb{Z}
\end{equation}
is a classical solution of the linearized AL equation (\ref{3.2}).
\label{lem-spectrum}
\end{lemma}

\begin{proof}
It follows from \eqref{3.4} that
\begin{equation}\label{lax-1-intro-3}
\left(\begin{array}{c}
P_{n-1}\\
Q_{n-1}\\
\end{array} \right) =  \frac{1}{\sqrt{1+ |U_{n-1}|^2}} \left(\begin{array}{cc}
\lambda^{-1} & -U_{n-1}\\
\bar{U}_{n-1} & \lambda \\
\end{array}
\right) \left(\begin{array}{c}
P_{n}\\
Q_{n}\\
\end{array} \right).
\end{equation}
By using \eqref{3.5} and (\ref{squared-eigen}), we obtain
\begin{align*}
i \dot{v}_n &= 2 \omega v_n - P_n^2(\lambda^3 + \lambda^{-1} + \lambda \bar{U}_n U_{n-1} + \lambda^{-1}
|U_n|^2) - Q_n^2(\lambda U_n^2 - \lambda^{-1} U_n U_{n-1})  \\
 & \quad + P_n Q_n \left[- 2 \lambda^2 U_n + U_{n-1} - U_{n+1} - |U_n|^2 (U_{n+1} + U_{n-1})\right]  \\
 & \quad+ \bar{P}_n^2 (\bar{\lambda}^{-1} U_n^2 - \bar{\lambda} U_n U_{n-1}) + \bar{Q}_n^2 \left(
\bar{\lambda} + \bar{\lambda}^{-3} + \bar{\lambda}^{-1} \bar{U}_n U_{n-1} + \bar{\lambda} |U_n|^2 \right) \\
& \quad + \bar{P}_n\bar{Q}_n \left[ -2 \bar{\lambda}^{-2} U_n + U_{n-1} - U_{n+1}  -
|U_n|^2 (U_{n+1} + U_{n-1}) \right].
\end{align*}
By using \eqref{3.4} and (\ref{lax-1-intro-3}), we obtain
\begin{align*}
(1 + |U_n|^2) v_{n+1} &=  P_n^2  (\lambda^3 - \lambda U_{n+1} \bar{U}_n ) +  P_n Q_n \left[
2 \lambda^2 U_n  + U_{n+1} (1 - |U_n|^2) \right]    \\
& \quad + Q_n^2 (\lambda U_n^2 + \lambda^{-1} U_n U_{n+1}) - \bar{P}_n^2 (\bar{\lambda}^{-1} U_n^2 +  \bar{\lambda} U_n U_{n+1}) \\
&  \quad + \bar{P}_n\bar{Q}_n \left[ 2 \bar{\lambda}^{-2} U_n +
U_{n+1} (1 - |U_n|^2) \right] + \bar{Q}_n^2 (\bar{\lambda}^{-1} U_{n+1} \bar{U}_n - \bar{\lambda}^{-3})
\end{align*}
and
\begin{align*}
v_{n-1} = \lambda^{-1} P_n^2 - \bar{\lambda} \bar{Q}_n^2 - U_{n-1} (P_n Q_n + \bar{P}_n\bar{Q}_n).
\end{align*}
By using \eqref{squared-eigen}, we obtain
\begin{align*}
U_n \bar{v}_n + \bar{U}_n v_n = \lambda \bar{U}_n P_n^2 - \lambda^{-1} U_n Q_n^2 + \bar{\lambda} U_n \bar{P}_n^2 - \bar{\lambda}^{-1} \bar{U}_n \bar{Q}_n^2 + 2 |U_n|^2 (P_n Q_n + \bar{P}_n\bar{Q}_n).
\end{align*}
When these expressions are substituted into the linearized AL equation \eqref{3.2}, all terms cancel out after direct computations.
\end{proof}

Two corollaries follow from the result of Lemma \ref{lem-spectrum}. 

\begin{corollary}
	\label{cor-squared-eigenfunction}
	If the linear system (\ref{3.4})--(\ref{3.5}) is solved with the separation of variables as
	\begin{equation}
	P_n(t) = \hat{P}_n e^{\Omega t}, \qquad Q_n(t) = \hat{Q}_n e^{\Omega t},  \label{eigenfunct-separ}
	\end{equation}
	where $\{ (\hat{P}_n, \hat{Q}_n)^{\mathrm{T}} \}$ is $t$-independent, 
	then the spectral stability problem (\ref{spec-stab}) is solved with
	\begin{equation}
	\label{Lambda-Omega-relation}
	V_n = \lambda \hat{P}_n^2 + U_n \hat{P}_n \hat{Q}_n, \quad
	\tilde{V}_n = -\lambda^{-1} \hat{Q}_n^2 + \bar{U}_n \hat{P}_n \hat{Q}_n, \quad
	\Lambda = 2 \Omega
	\end{equation}
	where $\tilde{V}_n$ is no longer a complex conjugate of $V_n$ if $\Lambda \notin \mathbb{R}$.
\end{corollary}

\begin{proof}
It follows from (\ref{squared-eigen}) that the squared eigenfunction relation yields
\begin{align*}
\left\{ \begin{array}{l}
v_n = \lambda P_n^2 - \bar{\lambda}^{-1} \bar{Q}_n^2 + U_n (P_n Q_n + \bar{P}_n \bar{Q}_n), \\
\bar{v}_n = \bar{\lambda} \bar{P}_n^2 - \lambda^{-1} Q_n^2 + \bar{U}_n (P_n Q_n + \bar{P}_n \bar{Q}_n).
\end{array}
\right.
\end{align*}
Substituting (\ref{eigenfunct-separ}) into these expressions yields a linear superposition of two solutions in the form 
$$
v_n(t) = V_n e^{\Lambda t}, \quad \bar{v}_n(t) = \tilde{V}_n e^{\Lambda t}.
$$
One solution is given by (\ref{Lambda-Omega-relation}) for $\Lambda = 2 \Omega$ and another solution is given by 
$$
V_n = -\bar{\lambda}^{-1} \bar{\hat{Q}}_n^2 + U_n \bar{\hat{P}}_n \bar{\hat{Q}}_n, \quad  \tilde{V}_n = \bar{\lambda} \bar{\hat{P}}_n^2 + \bar{U}_n \bar{\hat{P}}_n \bar{\hat{Q}}_n
$$ 
for $\Lambda = 2 \bar{\Omega}$.
\end{proof}

\begin{corollary}
	\label{lem-stab}
The spectral parameters $\Omega$ and $\lambda$
are related by the algebraic equation
\begin{equation}\label{3.6}
\Omega^2 + Q(\lambda) = 0,
\end{equation}
where
\begin{align*}
Q(\lambda) := \frac14 \left(\lambda^2 + \lambda^{-2} \right)^2 - \omega \left(\lambda^2 + \lambda^{-2} \right) + \omega^2 + \frac{i}{2}  F_0 \left( \lambda^2 - \lambda^{-2}\right) - \frac{1}{4} F_0^2 - F_1.
\end{align*}	
\end{corollary}

\begin{proof}
	After separation of variables with (\ref{eigenfunct-separ}), the time-evolution problem (\ref{3.5}) becomes a linear algebraic system, which admits a nonzero solution if and only if the determinant of the coefficient matrix is zero:
$$
\left| \begin{array}{cc}
W_n - \omega + i \Omega & \lambda U_n - \lambda^{-1} U_{n-1}\\
- \lambda \bar{U}_{n-1} + \lambda^{-1} \bar{U}_n & i \Omega + \omega - W_n \\
\end{array}
\right| = 0.
$$
Expanding the determinant and using the conserved quantities  (\ref{2.2-1}) and (\ref{2.2-2})
yields the algebraic equation for $\Omega$ in the form (\ref{3.6}).
\end{proof}

\begin{remark}
Corollaries \ref{cor-squared-eigenfunction} and \ref{lem-stab} relate the spectral parameters $\Lambda$, $\Omega$, and $\lambda$.
As a result, the stability spectrum $\Lambda$ of the spectral stability problem (\ref{spec-stab}) is fully determined in terms of the Lax spectrum $\lambda$ 
of the spectral problem (\ref{3.4}), for which eigenfunctions 
$\{  (\hat{P}_n,\hat{Q}_n)^{\mathrm{T}}\}_{n \in \mathbb{Z}}$ and 
$\{  (V_n,\tilde{V}_n)^{\mathrm{T}}\}_{n \in \mathbb{Z}}$
are required to be bounded in $n \in \mathbb{Z}$.
\end{remark}

\begin{remark}
Solving (\ref{3.6}) for $F_0 = 0$, we obtain 	
\begin{equation}
\label{Omega}
\Omega = \pm \frac{i}{2\lambda^2} \sqrt{P(\lambda)},
\end{equation}
where
\begin{align}
\label{polynomial-P}
P(\lambda) = 4 \lambda^4 Q(\lambda) =  \lambda^8 - 4 \omega \lambda^6 + 2 (1 + 2 \omega^2 - 2 F_1 ) \lambda^4	- 4 \omega \lambda^2 +1
\end{align}
is the same polynomial as in our previous work \cite{Chen-Pel-2022} up to the definition of $\omega$ and $F_1$. The polynomial $P(\lambda)$ was obtained in \cite{Chen-Pel-2022} by using the nonlinearization method, for which a certain relation between squared eigenfunctions $\{ (P_n,Q_n)^{\mathrm{T}} \}_{n \in \mathbb{Z}}$ of the linear system (\ref{3.4})--(\ref{3.5}) with some $\lambda = \lambda_1$ and the potential $\{ U_n \}_{n \in \mathbb{Z}}$ is imposed. After the relation is imposed, the Lax system becomes nonlinear and the potential satisfies a second-order difference equation (\ref{2.2}) which is satisfied by the periodic waves with trivial phase. Integrability of both  the nonlinear Lax system and the difference equation results in the construction of the polynomial $P(\lambda)$ with $\lambda_1$ being a root of $P(\lambda)$ \cite{Chen-Pel-2022}.
\end{remark}

\begin{remark}
	We were not able to recover the standing periodic waves by using the nonlinearization method for the Lax system (\ref{3.4})--(\ref{3.5}) associated with the AL equation (\ref{al}), which was developed in  \cite{3}. The relation between the squared eigenfunctions and the potential imposes constraints which do not recover the second-order difference equation (\ref{2.2}) for the standing periodic waves of the AL equation. Consequently, we have obtained the polynomial $P(\lambda)$ by using separation of variables for the standing periodic waves without relation to the nonlinearization method. Note that the separation of variables does not work for the discrete modified Korteweg--de Vries equation considered in \cite{Chen-Pel-2022}.
\end{remark}

The following lemma presents relations between the squared eigenfunctions $\{ (P_n,Q_n)^{\mathrm{T}} \}_{n \in \mathbb{Z}}$ of the linear system (\ref{3.4})--(\ref{3.5}) with $\lambda = \lambda_1$ being a root of $P(\lambda)$ and the potential $\{ U_n \}_{n \in \mathbb{Z}}$ satisfying  (\ref{2.2}), (\ref{2.2-1}), and  (\ref{2.2-2}) with $F_0 = 0$. In the case of the real-valued potentials (e.g., for the standing periodic waves with trivial phase), these relations recover those obtained in \cite{Chen-Pel-2022} with the nonlinearization method. Since the nonlinearization method does not work for the AL equation (\ref{al}), we have established these relations by using substitutions. 

\begin{lemma}
	Let $\{ U_n\}_{n \in \mathbb{Z}} \in \mathbb{C}^{\mathbb{Z}}$ be a solution of (\ref{2.2}), (\ref{2.2-1}), and (\ref{2.2-2}) with $F_0 = 0$. Let $\lambda_1 \in \mathbb{C}$ be a root of the polynomial $P(\lambda)$ in (\ref{polynomial-P}) and define
\begin{equation}
\label{omega-relation}
	\omega = \frac{1}{2} (\lambda_1^2 + \lambda_1^{-2}) + \sigma_1 \sqrt{F_1}
\end{equation}
with $\sigma_1 = \pm 1$.
Then, the eigenfunction $\{ (P_n,Q_n)^\mathrm{T} \}_{n \in \mathbb{Z}}$ of the linear system
	(\ref{3.4})--(\ref{3.5}) with $\lambda = \lambda_1$ is given
	up to a multiplicative constant by
	\begin{equation}
	\label{eigenfunction}
	\left\{ \begin{array}{l}
	P_n^2 = \lambda_1 U_n - \lambda_1^{-1} U_{n-1}, \\
	Q_n^2 = \lambda_1 \bar{U}_{n-1} - \lambda_1^{-1} \bar{U}_{n}, \\
	P_n Q_n = \sigma_1 \sqrt{F_1}-\frac{1}{2} (U_n \bar{U}_{n-1} + \bar{U}_n U_{n-1}).
	\end{array}  \right.
	\end{equation}
	\label{lem-relation}
\end{lemma}

\begin{proof}
Relation (\ref{omega-relation}) is found by solving $P(\lambda_1) = 0$ in
$\omega$ and picking one of the two squared roots.
Since the root of $P(\lambda)$ corresponds to $\Omega = 0$,
it follows from (\ref{3.5}) with $\lambda= \lambda_1$
that $P_n$ and $Q_n$ are related by
$$
\frac{1}{2} (U_n \bar{U}_{n-1} + \bar{U}_n U_{n-1} - 2 \sigma_1 \sqrt{F_1}) P_n + (\lambda_1 U_n - \lambda_1^{-1} U_{n-1}) Q_n = 0.
$$
Multiplying this relation by $P_n$ and by $Q_n$ verifies
the relations (\ref{eigenfunction}) in view of (\ref{2.2-1}) with $F_0 = 0$
and (\ref{2.2-2}). The relations (\ref{eigenfunction}) are compatible with the
spectral problem (\ref{3.4}), which results in the constraints
\begin{align*}
(1+|U_n|^2) P_{n+1}^2 &= \lambda_1^2 P_n^2 + 2 \lambda_1 U_n P_n Q_n + U_n^2 Q_n^2, \\
(1+|U_n|^2) Q_{n+1}^2 &= \bar{U}_n^2 P_n^2 - 2 \lambda_1^{-1} \bar{U}_n P_n Q_n + \lambda_1^{-2} Q_n^2 , \\
(1+|U_n|^2) P_{n+1} Q_{n+1} &= -\lambda_1 \bar{U}_n P_n^2 + (1-|U_n|^2) P_n Q_n + \lambda_1^{-1} U_n Q_n^2.
\end{align*}
Substituting (\ref{eigenfunction}) into these relations yields identities by using (\ref{2.2}) and (\ref{omega-relation}).
\end{proof}

\begin{remark}
	Both $U_n$ and $(P_n,Q_n)^\mathrm{T}$ in Lemma \ref{lem-relation} are independent of time $t$. Although the dnoidal and cnoidal waves have real-valued profile $\{ U_n \}_{n \in \mathbb{Z}} \in \mathbb{R}^{\mathbb{Z}}$, Lemma \ref{lem-relation} also holds for complex-valued profiles with $F_0 = 0$.
\end{remark}

\begin{remark}
The eigenfunction $\{ (P_n,Q_n)^\mathrm{T} \}_{n \in \mathbb{Z}}$ in Lemma \ref{lem-spectrum} is defined for an arbitrary value of $\lambda \in \mathbb{C}$. The eigenfunction $\{ (P_n,Q_n)^\mathrm{T} \}_{n \in \mathbb{Z}}$ in Lemma \ref{lem-relation} is defined for the root $\lambda = \lambda_1$ of the polynomial $P(\lambda)$ in (\ref{polynomial-P}). The latter eigenfunction generates solutions of the spectral stability problem (\ref{spec-stab}) by Corollaries \ref{cor-squared-eigenfunction} and \ref{lem-stab} with $\Lambda = 0$.
\end{remark}

\section{Lax and stability spectra for periodic waves with trivial phase}
\label{sec:4}

Here we construct the Lax spectrum and the stability spectrum for the constant-amplitude wave (\ref{2.4}), the dnoidal wave (\ref{2.3}), and the cnoidal wave (\ref{2.6}). The analytical result for the constant-amplitude wave (\ref{2.4}) is obtained with the band-limited Fourier transform. The stability results for the dnoidal and cnoidal waves are obtained with the assistance of numerical approximations of the Lax spectrum from the spectral problem 
(\ref{3.4}).

\subsection{Constant-amplitude wave}

We set $U_n = A$ with $A = \tan(\alpha)$ and $\omega = {\rm sec}^2(\alpha)$
for $\alpha \in (0,\frac{\pi}{2})$ as follows from (\ref{2.4}). The following lemma gives the exact location of the Lax spectrum and the stability spectrum for the constant-amplitude wave.

\begin{lemma}
	\label{lem-constant-wave} The Lax spectrum $\lambda$ of the spectral problem (\ref{3.4}) with $U_n = \tan(\alpha)$ and $\omega = {\rm sec}^2(\alpha)$ consists of the unit circle in $\mathbb{C}$ and two  segments on the real axis:
	$$
	[-\sqrt{\omega}-\sqrt{\omega - 1},-\sqrt{\omega} + \sqrt{\omega-1}] \cup
		[\sqrt{\omega}-\sqrt{\omega - 1},\sqrt{\omega} + \sqrt{\omega-1}].
	$$
	The stability spectrum $\Lambda$ of the spectral problem (\ref{spec-stab}) consists of the segment $[-2 (\omega-1), 2 (\omega-1)]$ on the real axis and the segment $[-4i\sqrt{\omega},4 i \sqrt{\omega}]$ on the purely imaginary axis.
\end{lemma}

\begin{proof}
We solve the spectral problem (\ref{3.4}) with constant $U_n = A$ by using the band-limited Fourier transform
$$
P_n = \frac{1}{2\pi} \int_{0}^{2 \pi} \hat{P}(\theta) e^{i \theta n} d \theta,
\quad
Q_n = \frac{1}{2\pi} \int_{0}^{2 \pi} \hat{Q}(\theta) e^{i \theta n} d \theta.
$$
A nontrivial solution for $(\hat{P}(\theta),\hat{Q}(\theta))^{\mathrm{T}}$ exists if and only if the following characteristic equation is satisfied:
$$
\left| \begin{array}{cc} \sqrt{1+A^2} e^{i \theta} - \lambda & -A \\ A & \sqrt{1+A^2} e^{i \theta} - \lambda^{-1} \end{array} \right| = 0.
$$	
This yields
$$
\sqrt{1+A^2} e^{i \theta} \left[ 2 \sqrt{1+A^2} \cos(\theta) - \lambda - \lambda^{-1} \right] = 0,
$$
which is equivalent to $z := \lambda + \lambda^{-1} = 2 \sqrt{1+A^2} \cos(\theta) = 2 \sqrt{\omega} \cos(\theta)$. When $\theta = 0$, two roots of the quadratic equation $\lambda + \lambda^{-1} = 2 \sqrt{\omega}$ are given by  $\lambda = \sqrt{\omega} + \sqrt{\omega - 1}$ and $\lambda = \sqrt{\omega} - \sqrt{\omega - 1}$. When $\theta$ changes in the interval $[0,2\pi]$, the two roots of the quadratic equation cover the two segments on the real axis and the unit circle in $\mathbb{C}$.

The stability spectrum (\ref{spec-stab}) with constant $U_n = A$ is written in the form:
\begin{equation*}
\left\{ \begin{array}{l}
i \Lambda V_n - 2 \omega V_n + (1 + A^2)(V_{n+1} + V_{n-1}) +
2 A^2(V_n + \tilde{V}_n) = 0, \\
-i \Lambda \tilde{V}_n - 2 \omega \tilde{V}_n + (1 + A^2)(\tilde{V}_{n+1} + \tilde{V}_{n-1}) +
2 A^2(V_n + \tilde{V}_n) = 0.
\end{array}
\right.
\end{equation*}
Using the band-limited Fourier transform with parameter $2\theta$ instead of $\theta$, we obtain a non-trivial solution for $(\hat{V}(\theta),\hat{\tilde{V}}(\theta))^{\mathrm{T}}$ if and only if the following characterstic equation is satisfied:
\begin{align*}
\left| \begin{array}{cc} 2 (1+A^2) \cos(2 \theta) - 2 \omega + 2 A^2 + i \Lambda & 2A^2 \\ 2A^2 &  2 (1+A^2) \cos(2 \theta) - 2 \omega + 2 A^2 - i \Lambda \end{array} \right| = 0,
\end{align*}
which implies for $\Lambda = 2 \Omega$ that
\begin{align}
\label{Omega-explicit}
\Omega = \pm i \sqrt{((1+A^2) \cos(2\theta) + A^2 - \omega)^2 - A^4}.
\end{align}
Since $\omega = 1 + A^2$, we obtain that
$$
\Omega = \pm 2i \sqrt{\omega} \sin(\theta) \sqrt{1 - \omega \cos^2(\theta)}.
$$
When $\theta \in [0,\arccos(\omega^{-1/2})]$, the values of $\Omega$ cover the segment $[-(\omega-1),(\omega-1)]$ on the real axis, where the maximal value is found from the extremal value of $\theta \mapsto \sin(\theta) \sqrt{\omega \cos^2(\theta) -1}$ located at $\theta_0 = \frac{1}{2} \arccos(\omega^{-1})$, for which
$$
\sin \theta_0 = \frac{\sqrt{\omega - 1}}{\sqrt{2\omega}}, \qquad
\cos \theta_0 = \frac{\sqrt{\omega + 1}}{\sqrt{2\omega}}.
$$
When $\theta \in [\arccos(\omega^{-1/2}),\frac{\pi}{2}]$, the values of $\Omega$ cover the segment $[-2 i \sqrt{\omega},2 i \sqrt{\omega}]$ on the purely imaginary axis, where the maximal value is attained at $\theta = \frac{\pi}{2}$. This yields the result for $\Lambda$ since $\Lambda = 2 \Omega$.
\end{proof}

\begin{remark}
	The explicit expression (\ref{Omega-explicit}) for $\Omega$ is consistent with the representation (\ref{Omega}) for $F_1 = A^4 = (1-\omega)^2$ and $z = \lambda + \lambda^{-1} = 2 \sqrt{\omega} \cos(\theta)$ since
\begin{align*}
\Omega &= \pm \frac{i}{2} \sqrt{\lambda^4 - 4 \omega \lambda^2 + 2(1+2\omega^2 - 2F_1) - 4 \omega \lambda^{-2} + \lambda^{-4}} \\
&= \pm \frac{i}{2} \sqrt{z^4 - 4(1+\omega) z^2 + 4(1 + 2\omega + \omega^2 - F_1)} \\
&= \pm 2i \sqrt{\omega} \sqrt{\omega \cos^4(\theta) - (1+\omega) \cos^2(\theta)+ 1} \\
&= \pm 2i \sqrt{\omega} \sin(\theta) \sqrt{1 - \omega \cos^2(\theta)},
\end{align*}
which recovers (\ref{Omega-explicit}).
\label{remark-mod-stab}
\end{remark}

\begin{remark}
	\label{remark-Omega-0}
It follows from Lemma \ref{lem-constant-wave} that
the constant-amplitude wave (\ref{2.4}) is modulationally unstable with the maximal growth rate being $\Lambda_0 = 2(\omega - 1) = 2 A^2$. The maximal growth rate coincides exactly with the maximal growth rate $\Lambda_0 = 2$ of the constant-amplitude wave $\mathfrak{u}(X,T) = e^{2iT}$ of the continuous NLS equation (\ref{nls}), where the parameter $A > 0$ is included in the formal scaling $\epsilon = A$.
\end{remark}

\subsection{Dnoidal wave}

We set $U_n = A {\rm dn}(\alpha n;k)$ according to (\ref{2.3})
with $A = {\rm sn}(\alpha;k)/{\rm cn}(\alpha;k)$
and look for nontrivial solutions of the spectral problem (\ref{3.4}) for which
$\{ (P_n,Q_n)^{\mathrm{T}} \}_{n \in \mathbb{Z}}$ is a bounded sequence as $|n| \to \infty$. Using the equivalent spectral problem (\ref{lax-1-intro-3}) with real $U_n$, we rewrite the spectral problem in the symmetric form:
\begin{equation}
\label{Floquet-spectrum}
\left\{ \begin{array}{l}
\sqrt{1+U_n^2} P_{n+1} + \sqrt{1+U_{n-1}^2} P_{n-1} - (U_n - U_{n-1}) Q_n = z P_n, \\
\sqrt{1+U_n^2} Q_{n+1} + \sqrt{1+U_{n-1}^2} Q_{n-1} + (U_n - U_{n-1}) P_n = z Q_n,
\end{array} \right.
\end{equation}
where $z := \lambda + \lambda^{-1}$.

The dnoidal function ${\rm dn}(\xi;k)$ has the period of $2 K(k)$. If $\alpha = K(k)/M$ for integer $M$,
then $U_{n+2M} = U_n$ so that we can use the discrete Floquet theory and set
\begin{equation}
\label{eigenmode}
P_n = \hat{P}_n(\theta) e^{i \theta n}, \quad Q_n = \hat{Q}_n(\theta) e^{i \theta n},
\end{equation}
with periodic $\hat{P}_{n+2M}(\theta) = \hat{P}_n(\theta)$, $\hat{Q}_{n+2M}(\theta) = \hat{Q}_n(\theta)$
for fixed Floquet parameter $\theta \in [0,\pi/M]$. Solving the $4M \times 4M$
matrix eigenvalue problem numerically gives the spectrum of $4M$ eigenvalues $z$ which are then traced by changing $\theta$ in $[0,\pi/M]$.

The Lax (Floquet) spectrum of the spectral problem (\ref{Floquet-spectrum})
for $k = 0.8$ and $M = 20$ is shown in Figure \ref{fig-1} on the $\lambda$-plane. The admissible values of $\lambda$ are found from invertion of $z = \lambda + \lambda^{-1}$ with two roots
\begin{equation}
\label{square-roots}
\lambda = \frac{z \pm \sqrt{z^2-4}}{2}.
\end{equation}
The red dots on the $\lambda$-plane show zeros $\{ \pm \lambda_1,\pm \lambda_1^{-1}, \pm \lambda_2, \pm \lambda_2^{-1} \}$ of the polynomial $P(\lambda)$ in (\ref{polynomial-P}) with
\begin{equation}
\label{3-1-0}
\omega = \frac{{\rm dn}(\alpha;k)}{{\rm cn}^2(\alpha;k)} \quad
\mbox{\rm and} \quad F_1 = (1-k^2) \frac{{\rm sn}^4(\alpha;k)}{{\rm cn}^4(\alpha;k)}.
\end{equation}
Applying the ordering $0 < \lambda_1 < \lambda_2 < 1$, the two roots $\lambda_1$ and $\lambda_2$ were obtained in \cite{Chen-Pel-2022}
in the explicit form:
\begin{equation}
\label{3-1-1}
\lambda_1 = \frac{\sqrt{\left(1 - {\rm sn}(\alpha;k)\right) \left({\rm dn}(\alpha;k) - \sqrt{1-k^2} {\rm sn}(\alpha;k)\right)}}{{\rm cn}(\alpha;k)}, \quad \sigma_1 = -1
\end{equation}
and
\begin{equation}
\label{3-1-2}
\lambda_2 = \frac{\sqrt{\left(1 - {\rm sn}(\alpha;k)\right) \left({\rm dn}(\alpha;k) + \sqrt{1-k^2} {\rm sn}(\alpha;k) \right)}}{{\rm cn}(\alpha;k)}, \quad \sigma_1 = +1.
\end{equation}
The choice of $\sigma_1$ refers to the relation (\ref{omega-relation})
when either $\lambda_1$ or $\lambda_2$ are taken in place of $\lambda_1$.
The left panel suggests that the Lax spectrum consists of the unit circle and four segments on the real axis between the eight real roots of $P(\lambda)$. The right panel shows the zoomed version to indicate that the spectral bands are connected between the roots.

\begin{figure}[htbp!]
	\includegraphics[width=8cm,height=6cm]{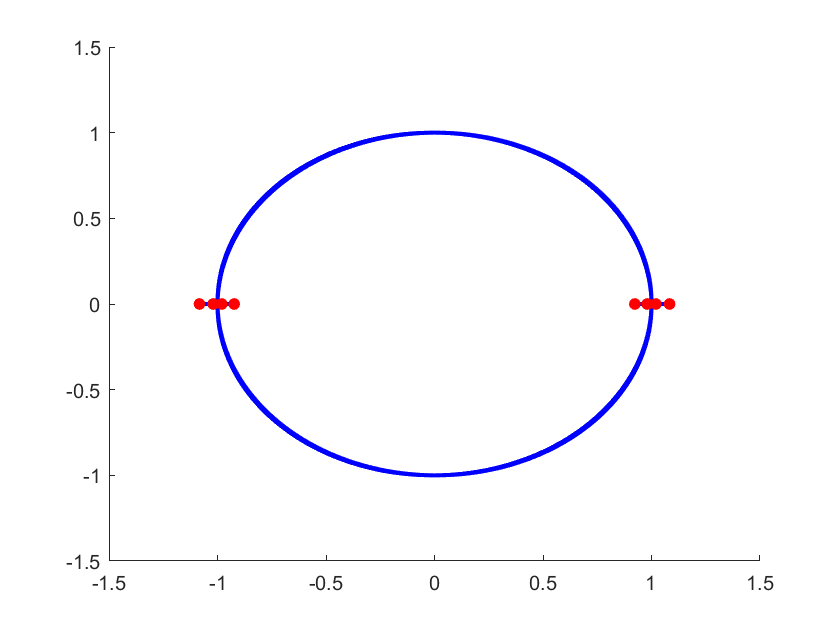}
	\includegraphics[width=8cm,height=6cm]{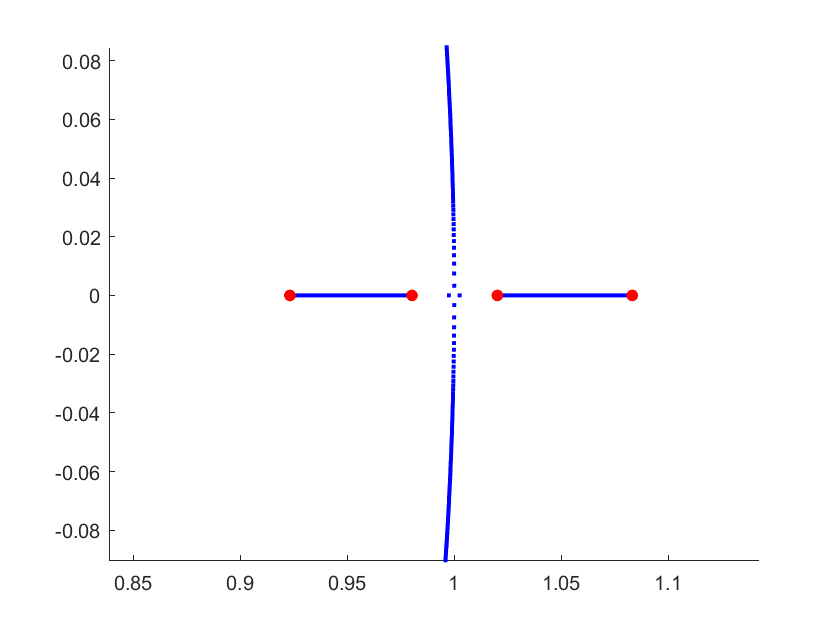}
	\caption{Lax spectrum (left) and its zoom (right) for the dnoidal wave with $\alpha = K(k)/M$ with $M = 20$ for $k = 0.8$. Red dots show roots of $P(\lambda)$.}
	\label{fig-1}
\end{figure}

The left panel of Figure \ref{fig-1a} shows the stability spectrum in the $\Omega$ plane, where we recall that $\Lambda$ in the spectral stability problem (\ref{spec-stab}) is given by $\Lambda = 2 \Omega$ in Corollary \ref{cor-squared-eigenfunction}. We relate the values of $\Omega$ and $z$ by
using (\ref{Omega})--(\ref{polynomial-P}) with the substituion $z^2 = \lambda^2 + \lambda^{-2} + 2$. This yields the explicit expression
\begin{equation}
\label{square-roots-Omega}
\Omega = \pm \frac{i}{2} \sqrt{z^4 - 4(1+\omega) z^2 + 4(1 + 2\omega + \omega^2 - F_1)}.
\end{equation}
The values of $z$ are obtained from the Lax spectrum computed numerically from the spectral problem (\ref{Floquet-spectrum}). Since the stability spectrum include the band on the positive real axis of $\Omega$, the dnoidal wave is
modulationally unstable similar to the constant-amplitude wave in Remark \ref{remark-mod-stab}.

The right panel of Figure \ref{fig-1a} shows $\Gamma := \Omega_0/\alpha^2$ versus $\alpha$, where $\Omega_0$ is the maximal positive $\Omega$ of the stability spectrum in the $\Omega$ plane. The numerically obtained values are shown by blue dots. The red dashed line shows the constant value $\sqrt{1-k^2}$
which is the maximal positive $\Omega$ of the dnoidal waves
in the continuous NLS equation (\ref{nls}) \cite{CPW}. As $k \to 0$,
this limiting value agrees with the value $\Omega_0 = \Lambda_0/2 = 1$ in Remark \ref{remark-Omega-0}.

\begin{figure}[htbp!]
	\includegraphics[width=8cm,height=6cm]{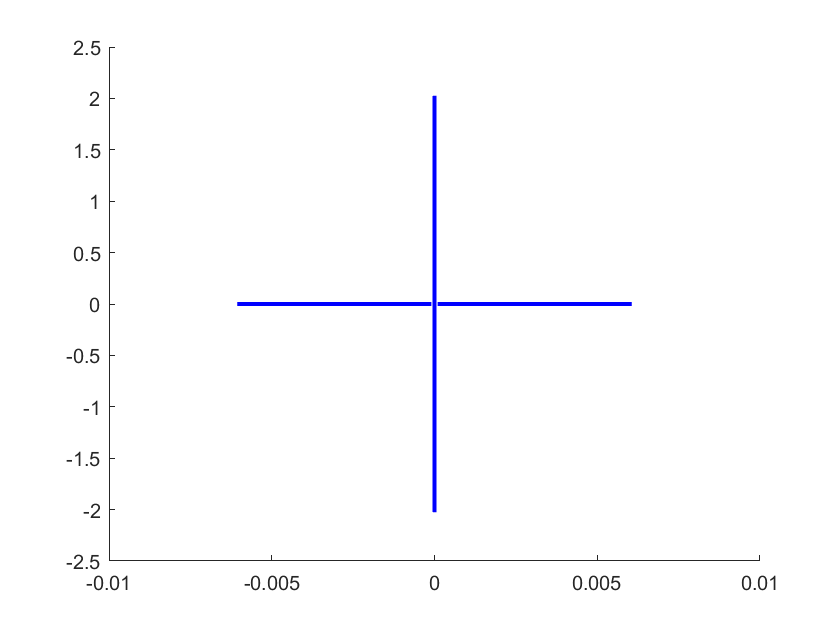}
	\includegraphics[width=8cm,height=6cm]{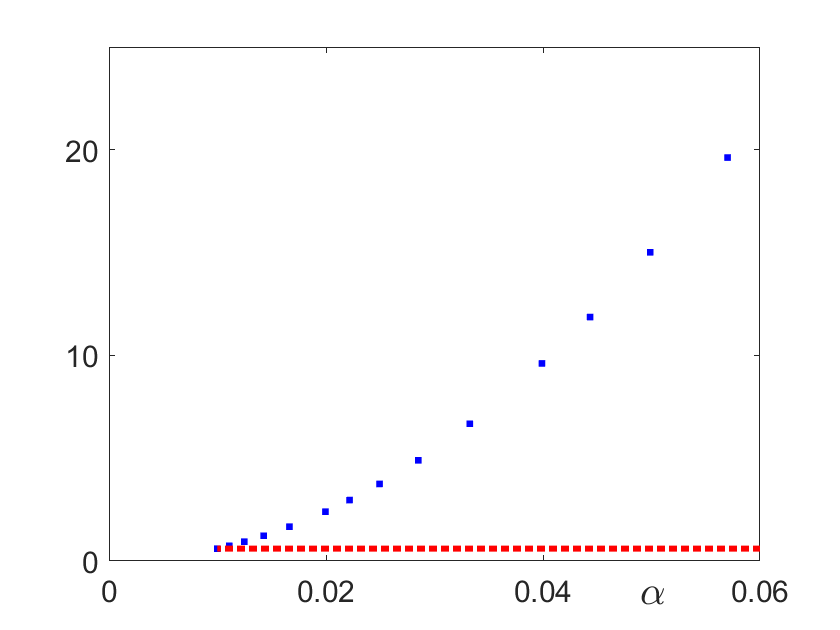}
	\caption{Left: Stability spectrum in the $\Omega$ plane for the dnoidal wave with $\alpha = K(k)/M$ with $M = 20$ for $k = 0.8$. Right: $\Omega_0/\alpha^2$ versus parameter $\alpha$ for fixed $k = 0.8$ (dots) and the constant value $\sqrt{1-k^2}$ (dotted line).}
	\label{fig-1a}
\end{figure}

\subsection{Cnoidal wave}

We set $U_n = A {\rm cn}(\alpha n;k)$ according to (\ref{2.6})
with $A = k {\rm sn}(\alpha;k)/{\rm dn}(\alpha;k)$
and look for bounded solutions of the spectral problem (\ref{Floquet-spectrum}) in the form (\ref{eigenmode}). Figure \ref{fig-2} show the Lax spectrum (left) and the stability spectrum (right) for the cnoidal wave with $k = 0.8$ (top) and $k = 0.95$ (bottom). The red dots on the $\lambda$-plane show again zeros $\{ \pm \lambda_1,
\pm \lambda_1^{-1}, \pm \bar{\lambda}_1,  \pm \bar{\lambda}_1^{-1} \}$ of the polynomial $P(\lambda)$ in (\ref{polynomial-P}) with
\begin{equation} \label{3-2-2-0}
\omega = \frac{{\rm cn}(\alpha;k)}{{\rm dn}^2(\alpha;k)}, \quad
 \quad F_1 = -k^2 (1-k^2) \frac{{\rm sn}^4(\alpha;k)}{{\rm dn}^4(\alpha;k)}.
\end{equation}
The complex value of $\lambda_1$ was obtained in \cite{Chen-Pel-2022} in the explicit form:
\begin{equation}\label{3-2-2}
\lambda_1 = \frac{\sqrt{(1 - k {\rm sn}(\alpha;k))({\rm cn}(\alpha;k) + i \sqrt{1-k^2} {\rm sn}(\alpha;k))}}{{\rm dn}(\alpha;k)}, \quad \sigma_1 = +1.
\end{equation}

\begin{figure}[htbp]
	\includegraphics[width=8cm,height=6cm]{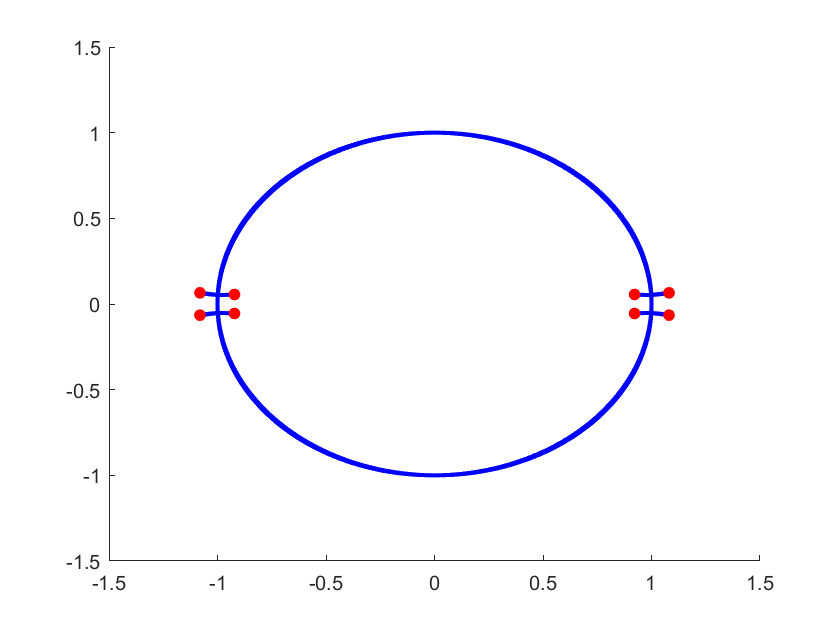}		\includegraphics[width=8cm,height=6cm]{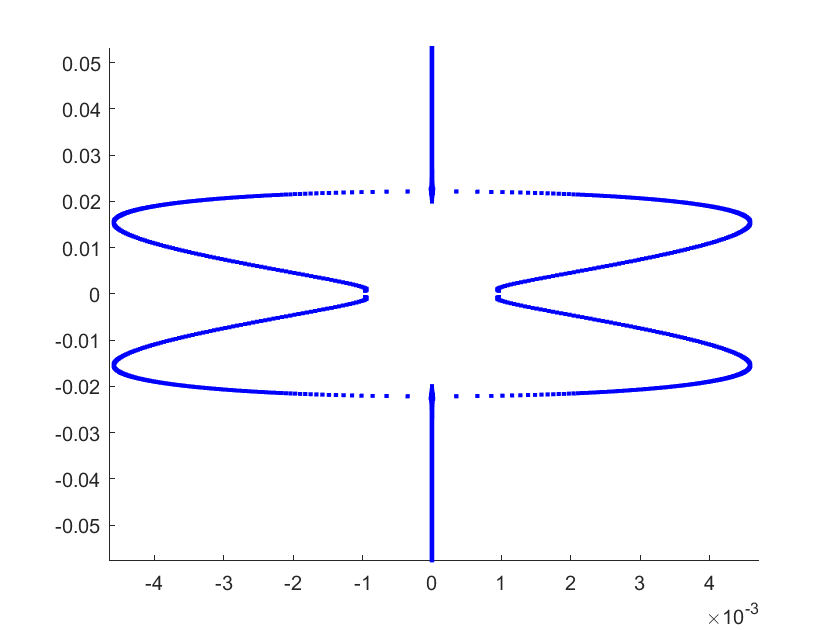}
	\includegraphics[width=8cm,height=6cm]{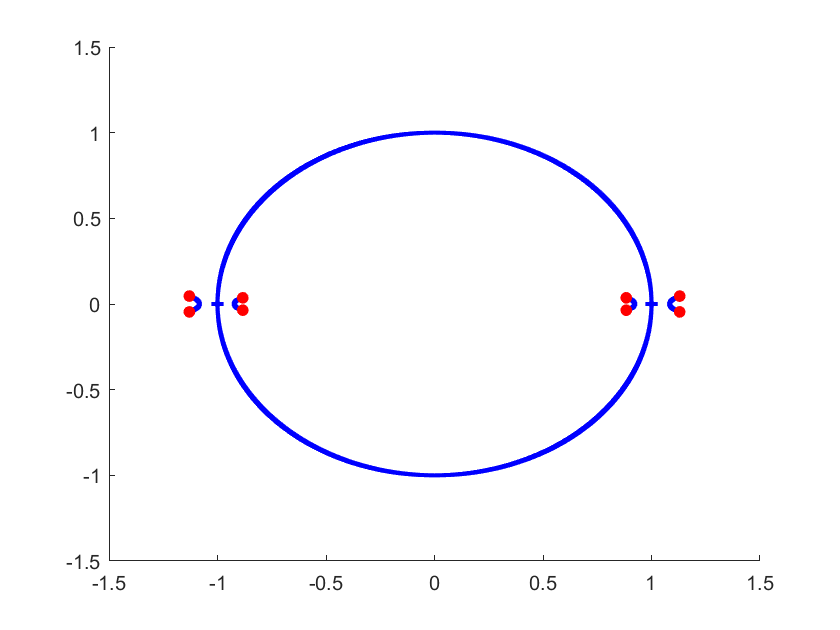}		\includegraphics[width=8cm,height=6cm]{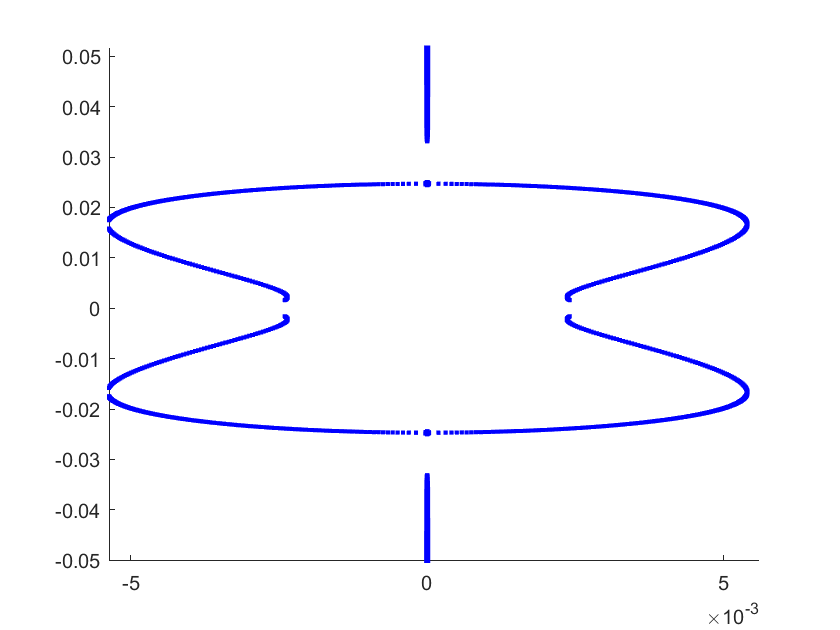}
	\caption{Lax spectrum (left) and stability spectrum on the $\Omega$ plane (right) for the cnoidal wave with $\alpha = K(k)/M$ with $M = 20$ for $k = 0.8$ (top) and $k = 0.95$ (bottom).}
	\label{fig-2}
\end{figure}

The numerical results suggest that the Lax spectrum consists of the unit circle and four complex bands which are either connected across the unit circle for $k = 0.8$ or connected away from the unit circle for $k = 0.95$. The stability spectrum in both cases displays the figure-eight instability bands in addition to two finite segments along the imaginary axis, boundary of which are beyond the margins of the right panels. The segments on the imaginary axis intersect the figure-eight bands for $k = 0.8$ and avoid intersection for $k = 0.95$. Because the figure-eight instability bands intersect the origin, the cnoidal wave is also declared to be modulationally unstable with a different instability pattern compared to the dnoidal wave.

We admit that the resolution of numerical data is poor near the origin on the right panels of Figure \ref{fig-2} because resolution of the Lax spectrum is poor near the points $\pm 1$ on the left panels. It is likely that sensitivity of numerical detected eigenvalues is related to evaluting the square roots in (\ref{square-roots}) and (\ref{square-roots-Omega}) near $z = \pm 2$.

\section{Nonperiodic solutions of linear equaitons}
\label{sec:6}

Since the dnoidal and cnoidal waves are modulationally unstable, as is shown in Section \ref{sec:4} based on numerical approximations of the Lax spectrum, we expect the existence of rogue waves (spatially and temporally localized solutions) on the modulationally unstable backgroud. Such solutions were obtained for the continuous NLS equation in \cite{CPnls,CPW,Feng} and for other related continuous equations in \cite{Chen-JNLS,Chen-DNLS,Ling,PW}. 

In order to obtain the rogue waves, as we do in Section \ref{sec:5} by using the analytical theory, we will consider here eigenfunctions of the linear system (\ref{3.4})--(\ref{3.5}) for the eigenvalue $\lambda = \lambda_1$ given by a root of the polynomial $P(\lambda)$ in (\ref{polynomial-P}). Since any root of $P(\lambda)$ is suitable, this notation for $\lambda_1$ is abstract and is not related to the particular choices in (\ref{3-1-1}), (\ref{3-1-2}), or (\ref{3-2-2}). 

According to the relations (\ref{eigenfunction}), the squared eigenfunctions 
corresponding to the eigenvalue $\lambda = \lambda_1$ are periodic. In addition to the periodic eigenfunctions, we construct here the second solution of the linear equations (\ref{3.4})--(\ref{3.5}) which are unbounded both in $n \in \mathbb{Z}$ and $t \in \mathbb{R}$.

The following lemma gives a construction of the second solution of the linear equations (\ref{3.4})--(\ref{3.5}) for an eigenvalue $\lambda = \lambda_1$.

\begin{lemma}
	\label{lem-second-solution}
	Let $\{ (P_n(t), Q_n(t) )^{\mathrm{T}} \}_{n \in \mathbb{Z}}$ be the eigenfunction of the linear system (\ref{3.4})--(\ref{3.5}) for an eigenvalue $\lambda = \lambda_1$. The second solution $\{ (\hat{P}_n(t), \hat{Q}_n(t) )^{\mathrm{T}} \}_{n \in \mathbb{Z}}$  of the linear system (\ref{3.4})--(\ref{3.5}) for the same eigenvalue $\lambda = \lambda_1$ can be represented in the form
	\begin{equation}
	\label{eigenfunction-second}
	\hat{P}_n(t) = P_n(t) \theta_n(t) - \frac{\bar{Q}_n(t)}{|P_n(t)|^2 + |Q_n(t)|^2}, \quad
\hat{Q}_n(t) = Q_n(t) \theta_n(t) + \frac{\bar{P}_n(t)}{|P_n(t)|^2 + |Q_n(t)|^2},
	\end{equation}
where $\theta_n(t)$ is a solution of the linear equations
\begin{equation}
\label{theta-1}
\theta_{n+1}-\theta_n = \frac{(|\lambda_1|^2 - 1) (\bar{\lambda}_1 U_n \bar{P}_n^2 - \lambda_1 \bar{U}_n \bar{Q}_n^2 - (1 + |\lambda_1|^2) \bar{P}_n \bar{Q}_n)}{(|P_n|^2 + |Q_n|^2) \Delta_n}
\end{equation}
and
\begin{align}
\frac{d \theta_n}{dt}  &= \frac{i (|\lambda_1|^2 - 1) \Sigma_n}{|\lambda_1|^2 (|P_n|^2 + |Q_n|^2)^2}
\label{theta-2}
\end{align}
with
\begin{align*}
\Delta_n &:= |\lambda_1|^4 |P_n|^2 + |Q_n|^2 +|\lambda_1|^2 |U_n|^2 (|P_n|^2 + |Q_n|^2) + (|\lambda_1|^2 - 1) ( \bar{\lambda}_1 U_n \bar{P}_n Q_n + \lambda_1 \bar{U}_n P_n \bar{Q}_n), \\
\Sigma_n & := (\lambda_1 U_n
+ \bar{\lambda}_1 U_{n-1}) \bar{P}_n^2 + (\bar{\lambda}_1 \bar{U}_n + \lambda_1 \bar{U}_{n-1}) \bar{Q}_n^2 - (1 + |\lambda_1|^{-2}) (\lambda_1^2 - \bar{\lambda}_1^2) \bar{P}_n\bar{Q}_n.
\end{align*}
\end{lemma}

\begin{proof}
	We obtain from (\ref{3.4}) that
	\begin{align}
	\label{tech-eq-1}
	|P_{n+1}|^2 + |Q_{n+1}|^2 = \frac{|\lambda_1 P_n + U_n Q_n|^2 + |\lambda_1^{-1} Q_n - \bar{U}_n P_n|^2}{1 + |U_n|^2}.
	\end{align}
	If $(P_n,Q_n)^{\mathrm{T}}$ and  $(\hat{P}_n,\hat{Q}_n)^{\mathrm{T}}$
	satisfy the linear equation (\ref{3.4}), then
	expression (\ref{eigenfunction-second}) implies that $\{ \theta_n \}_{n \in \mathbb{Z}}$ is a solution of
	\begin{align*}
	(\lambda_1 P_n + U_n Q_n) (\theta_{n+1} - \theta_n) =
	\frac{ U_n \bar{P}_n - \lambda_1 \bar{Q}_n}{|P_n|^2 + |Q_n|^2} -
\frac{ U_n \bar{P}_n - \bar{\lambda}_1^{-1} \bar{Q}_n}{|P_{n+1}|^2 + |Q_{n+1}|^2}.
	\end{align*}
Substituting (\ref{tech-eq-1}) and dividing by 	$(\lambda_1 P_n + U_n Q_n)$
yield (\ref{theta-1}) after simplifications.

We obtain from \eqref{3.5} that
\begin{align}
\frac{d}{d t} (|P_n|^2 + |Q_n|^2) &=  i (W_n - \bar{W}_n)
(|P_n|^2 - |Q_n|^2) \nonumber \\ &
\quad  + i \left[(\lambda_1 - \bar{\lambda}_1^{-1})U_n + (\bar{\lambda}_1 - {\lambda}_1^{-1})U_{n-1} \right] \bar{P}_n Q_n \nonumber \\
& \quad
 - i \left[ (\bar{\lambda}_1 - {\lambda}_1^{-1})\bar{U}_n + (\lambda_1 - \bar{\lambda}_1^{-1}) \bar{U}_{n-1} \right] P_n \bar{Q}_n.
\label{tech-eq-2}
\end{align}
If $(P_n,Q_n)^{\mathrm{T}}$ and  $(\hat{P}_n,\hat{Q}_n)^{\mathrm{T}}$
satisfy the linear equation (\ref{3.5}), then substituting \eqref{eigenfunction-second} and \eqref{tech-eq-2} into (\ref{3.5})
	and dividing by $P_n$ yield \eqref{theta-2} after simplifications.
\end{proof}

\begin{remark}
	\label{rem-reduction}
	The result of Lemma \ref{lem-second-solution} does not use the relations (\ref{eigenfunction}). In other words, $\lambda_1$ in Lemma \ref{lem-second-solution} does not have to be a root of $P(\lambda)$ in (\ref{polynomial-P}).
\end{remark}

If we use the relations (\ref{eigenfunction}) as in Remark \ref{rem-reduction},
then we can simplify the relations (\ref{theta-1}) and (\ref{theta-2}).
This is done separately for the case of dnoidal and cnoidal waves.
The following lemma summarizes the results of these computations.

\begin{lemma}
	Let $\{ U_n \}_{n \in \mathbb{Z}}$ be either the dnoidal or cnoidal waves given by (\ref{2.3}) or (\ref{2.6}) and $\lambda_1$ be a root of the polynomial $P(\lambda)$ in (\ref{polynomial-P}). Then, $\theta_n(t) = \Theta_n + i t$ in the representation (\ref{eigenfunction-second}) with $\{ \Theta_n\}_{n \in \mathbb{Z}}$ being a time-independent solution of the difference equations:
	\begin{equation}
	\label{theta-1-1}
	\theta_{n+1}-\theta_n = \frac{(\lambda_1 + {\lambda}_1^{-1}) (|U_n|^2 - \sigma_1 \sqrt{F_1})}{(\lambda_1 - \lambda_1^{-1}) (F_1 + 2(1 + \omega) |U_n|^2 + |U_n|^4)}
	\end{equation}
	or
	\begin{equation}
	\theta_{n+1}-\theta_n = \frac{(|\lambda_1|^2 - 1) (\bar{\lambda}_1 \lambda_1^{-1} + \bar{\lambda}_1^{-2}) |U_n|^2 + \sqrt{F_1}
		(|\lambda_1|^{2} - |\lambda_1|^{-2})}{ \Gamma_n},
	\label{theta-2-7}
	\end{equation}
	with
	\begin{align*}
	\Gamma_n
	&= |U_n|^2 \Big( |\lambda_1|^{4} + |\lambda_1|^{-4} + 2 |U_{n-1}|^2 - 2 |U_n|^2 \Big) +  2 |U_{n-1}|^2 \\
	& \quad + (|\lambda_1|^{2} + |\lambda_1|^{-2}) \Big( |U_n|^4 - F_1 - \bar{\lambda}_1 \lambda_1^{-1} \bar{U}_n U_{n-1} - \lambda_1 \bar{\lambda}_1^{-1} U_n \bar{U}_{n-1} \Big)
	\end{align*}
where (\ref{theta-1-1}) and (\ref{theta-2-7}) correspond to the dnoidal or cnoidal waves respectively.
	\label{lem-clear-expressions}
\end{lemma}

\begin{proof}
For the cnoidal wave, we assume that $\lambda_1 \in \mathbb{C}$ is given by (\ref{3-2-2}), whereas $\omega > 0$ and $F_1 < 0$ are given by (\ref{3-2-2-0}). It follows from \eqref{omega-relation} with $\sigma_1 = 1$ and $F_1 < 0$ that
	\begin{equation}\begin{array}{l}
	\lambda_1^2 + \lambda_1^{-2} + \bar{\lambda}_1^2 + \bar{\lambda}_1^{-2} = 4 \omega,\\
	\lambda_1^2 + \lambda_1^{-2} - \bar{\lambda}_1^2 - \bar{\lambda}_1^{-2} = - 4 \sqrt{F_1}.\\
	\end{array}
	\label{theta-2-6}
	\end{equation}
It follows from  \eqref{eigenfunction} with $\sigma_1 = 1$, $F_0 = 0$, and $F_1 < 0$ that
\begin{align*}
(|P_n|^2 + |Q_n|^2)^2 &= |\lambda_1 U_n - \lambda_1^{-1} U_{n-1}|^2 + |\lambda_1 \bar{U}_{n-1} - \lambda_1^{-1} \bar{U}_n|^2 + \frac12 (U_n \bar{U}_{n-1} + \bar{U}_n U_{n-1})^2 - 2 F_1 \\
 &= (|\lambda_1|^2 + |\lambda_1|^{-2}) (|U_n|^2 + |U_{n-1}|^2) - 2\lambda_1^{-1} \bar{\lambda}_1 \bar{U}_n U_{n-1} - 2 \lambda_1 \bar{\lambda}_1^{-1} U_n \bar{U}_{n-1} \\ &
\quad + 2 |U_n|^2 |U_{n-1}|^2 - 2 F_1.
\end{align*}
Similarly, we obtain from (\ref{eigenfunction}) that
\begin{align*}
(1-|\lambda_1|^{-2}) \Sigma_n  &= (|\lambda_1|^2 - 2 + |\lambda_1|^{-2})
(|U_n|^2 + |U_{n-1}|^2) +2 \bar{\lambda}_1 (\bar{\lambda}_1 - \lambda_1^{-1}) \bar{U}_n U_{n-1} \\
& \quad  + 2 \bar{\lambda}_1^{-1} (\bar{\lambda}_1^{-1} - \lambda_1) U_n \bar{U}_{n-1} + \frac{1}{2} (\lambda_1^2 + \lambda_1^{-2} - \bar{\lambda}_1^2 - \bar{\lambda}_1^{-2}) (U_n \bar{U}_{n-1} + \bar{U}_n U_{n-1} + 2 \sqrt{F_1}),
\end{align*}
where we have used that $\sqrt{F_1} \in i \mathbb{R}$. By using $F_0 = 0$ and (\ref{theta-2-6}), we can simplify the previous expression to the form:
\begin{align*}
(1-|\lambda_1|^{-2}) \Sigma_n  &= (|\lambda_1|^2 - 2 + |\lambda_1|^{-2})
(|U_n|^2 + |U_{n-1}|^2) -2 \bar{\lambda}_1 \lambda_1^{-1} \bar{U}_n U_{n-1} - 2 \bar{\lambda}_1^{-1} \lambda_1 U_n \bar{U}_{n-1} \\
& \quad + \frac{1}{2} (\lambda_1^2 + \lambda_1^{-2} - \bar{\lambda}_1^2 - \bar{\lambda}_1^{-2}) (U_n \bar{U}_{n-1} + \bar{U}_n U_{n-1} + 2 \sqrt{F_1}) \\
&= (|\lambda_1|^2 - 2 + |\lambda_1|^{-2})
(|U_n|^2 + |U_{n-1}|^2) -2 \bar{\lambda}_1 \lambda_1^{-1} \bar{U}_n U_{n-1} - 2 \bar{\lambda}_1^{-1} \lambda_1 U_n \bar{U}_{n-1} \\
& \quad + 2 \omega (U_n \bar{U}_{n-1} + \bar{U}_n U_{n-1}) - 4 F_1.
\end{align*}
Due to the conservation (\ref{2.2-2}), this implies that
$(1-|\lambda_1|^{-2}) \Sigma_n = (|P_n|^2 + |Q_n|^2)^2$ and hence
(\ref{theta-2}) becomes a trivial equation $\dot{\theta}_n = i$,
with the solution $\theta_n(t) = \Theta_n + i t$, where $\Theta_n$ is $t$-independent. To get $\{ \Theta_n \}_{n \in \mathbb{Z}}$, we substitute \eqref{eigenfunction} into \eqref{theta-1} and simplify the result with a lengthy but direct computation to the form (\ref{theta-2-7}), where we have used that $\sigma_1 = 1$, $F_0 = 0$, and $\sqrt{F_1} \in i \mathbb{R}$.

For the dnoidal wave, we assume that $\lambda_1 \in \mathbb{R}$ is given by either (\ref{3-1-1}) or (\ref{3-1-2}), whereas $\omega > 0$ and $F_1 > 0$ are given by (\ref{3-1-0}). Without loss of generality, we consider (\ref{3-1-1}) with $\sigma_1 = -1$. It follows from  \eqref{eigenfunction} with $\sigma_1 = -1$ that
\begin{align*}
(|P_n|^2 + |Q_n|^2)^2 &= ({\lambda}_1^2 + {\lambda}_1^{-2}) (|U_n|^2 + |U_{n-1}|^2) + \frac12
(U_n \bar{U}_{n-1} + \bar{U}_n U_{n-1} + 2 \sqrt{F_1})^2 \\ & \quad \quad 
- 2 (U_n \bar{U}_{n-1} + \bar{U}_n U_{n-1}) \\
&= (\lambda_1^2 + \lambda_1^{-2} - 2) (\bar{U}_n U_{n-1} + U_n \bar{U}_{n-1}
+ |U_n|^2 + |U_{n-1}|^2)
\end{align*}
and
\begin{align*}
(1- \lambda_1^{-2}) \Sigma_n &= ({\lambda}_1^2 + \lambda_1^{-2} - 2) (|U_n|^2 + |U_{n-1}|^2 ) + 2( \lambda_1^2 - 1) \bar{U}_n U_{n-1} - 2 (1 - \lambda_1^{-2}) U_n \bar{U}_{n-1} \\
&= (\lambda_1^2 + \lambda_1^{-2} - 2) (\bar{U}_n U_{n-1} + U_n \bar{U}_{n-1}
+ |U_n|^2 + |U_{n-1}|^2),
\end{align*}
where we have used \eqref{2.2-1} and \eqref{2.2-2} with $F_0 = 0$ and $F_1 > 0$ and (\ref{omega-relation}) with $\sigma_1 = -1$. Thus, we have $(1-\lambda_1^{-2}) \Sigma = (|P_n|^2 + |Q_n|^2)^2$ so that
(\ref{theta-2}) becomes again a trivial equation $\dot{\theta}_n = i$. Again, we have $\theta_n(t) = \Theta_n + i t$, where the $t$-independent $\Theta_n$ is obtained from the difference equation \eqref{theta-1}. After long but straightforward computations, we have simplified the expression to the form (\ref{theta-1-1}), where we have used (\ref{omega-relation}) and restored the value of $\sigma_1$ for either (\ref{3-1-1}) or (\ref{3-1-2}).
\end{proof}

\section{Construction of rogue waves}
\label{sec:5}

Here we construct the rogue waves on the background of the standing periodic waves. To do so, we use the one-fold Darboux transformation ($1$-fold DT).
The $1$-fold DT has already been constructed for the AL equation (\ref{al})
in \cite{LiuZeng} but the formulas used eigenfunctions of the Lax system
(\ref{lax-1-intro})--(\ref{lax-2-intro}), which is not suitable for the standing periodic waves as in Remark \ref{rem-lax}. Therefore, our first task is to extend the $1$-fold DT to the eigenfunctions of the Lax system (\ref{lax-intro}). We achieve the task with direct computations similarly to
computations in \cite{XuPelin} and \cite{Chen-Pel-2022}.

The following lemma presents the $1$-fold DT for the AL equation (\ref{al}) in terms of solutions of the linear system (\ref{lax-1-intro}).

\begin{lemma}
	\label{lem-one-fold}
	Let $\{ u_n(t)\}_{n \in \mathbb{Z}}$ be a solution of the AL equation \eqref{al}, $\{ (p_n(t), q_n(t))^\mathrm{T} \}_{n \in \mathbb{Z}}$ be a nontrivial solution to the linear system \eqref{lax-intro} with $\lambda = \lambda_1$, and $\{ \varphi_n(t) \}_{n \in \mathbb{Z}}$ be any solution to the linear system \eqref{lax-intro} with arbitrary $\lambda \in \mathbb{C}$. Then, $\{ \hat{u}_n(t)\}_{n \in \mathbb{Z}}$ given by 
	\begin{equation}\label{5.1}
	\hat{u}_n = - \frac{\lambda_1(|p_n|^2 + |\lambda_1|^2 |q_n|^2)}{\bar{\lambda}_1(|\lambda_1|^2 |p_n|^2 + |q_n|^2)} u_n + \frac{\lambda_1 (1 - |\lambda_1|^4) p_n \bar{q}_n}{\bar{\lambda}_1^2 (|\lambda_1|^2 |p_n|^2 + |q_n|^2)}
	\end{equation}
	is a new solution of the AL equation \eqref{al} and $\{ \hat{\varphi}_n(t) \}_{n \in \mathbb{Z}}$ given by 
	$\hat{\varphi}_n = M_n(\lambda) \varphi_n$ is a new solution to the linear system 
	\eqref{lax-intro} with arbitrary $\lambda$, where
	\begin{equation*}
	M_n(\lambda) = \frac{\sqrt{|p_n|^2 + |\lambda_1|^2 |q_n|^2}}{\sqrt{|\lambda_1|^2 |p_n|^2 + |q_n|^2}} \left(\begin{array}{cc} \lambda + \lambda^{-1} a_n & b_n \\
	-\bar{b}_n & \lambda \bar{a}_n + \lambda^{-1} \end{array} \right)
	\end{equation*}
	with
	\begin{equation*}
	a_n = - \frac{\lambda_1(|\lambda_1|^2 |p_n|^2 + |q_n|^2)}{\bar{\lambda}_1(|p_n|^2 + |\lambda_1|^2 |q_n|^2)}, \qquad b_n =
	\frac{\lambda_1(1 - |\lambda_1|^4) p_n \bar{q}_n}{|\lambda_1|^2 (|p_n|^2 + |\lambda_1|^2 |q_n|^2)}.
	\end{equation*}
\end{lemma}

\begin{proof}
	We need to show validity of the Darboux equations
	\begin{equation}
	\label{5.4}
	{U}(\hat{u}_n, \lambda) M_n(\lambda) = M_{n+1}(\lambda) {U}(u_n,\lambda)
	\end{equation}
	and
	\begin{equation}
	\label{5.5}
	V(\hat{u}_n, \lambda) M_n(\lambda) = \dot{M}_n(\lambda) + M_n(\lambda) V(u_n,\lambda).
	\end{equation}
	Substituting $U(u_n,\lambda)$ and $M_n(\lambda)$ into \eqref{5.4} and collecting different powers with respect to $\lambda$ yields the system of equations
	\begin{equation}
	\label{5.6}
	\left\{ \begin{array}{l}
	\bar{a}_n \hat{u}_n + b_n - u_n = 0, \\
	\hat{u}_n - a_{n+1} u_n - b_{n+1} = 0, \\
	a_{n+1} - a_n - b_{n+1} \bar{u}_n + \bar{b}_n \hat{u}_n = 0, \\
	a_{n+1} (1+u_n^2) - a_n (1+ \hat{u}_n^2) = 0.  \end{array} \right.
	\end{equation}
It follows from $\varphi_{n+1} = U(u_n,\lambda) \varphi_n$ that
	\begin{align*}
	|p_{n+1}|^2 + |\lambda_1|^2 |q_{n+1}|^2 = |\lambda_1|^2 |p_n|^2 + |q_n|^2,
	\end{align*}
	\begin{align*}
	(1 + |u_n|^2) (|\lambda_1|^2 |p_{n+1}|^2 + |q_{n+1}|^2) &= 
	(|\lambda_1|^2 + |u_n|^2) |p_n|^2 +
	u_n \bar{p}_n q_n (\bar{\lambda}_1 |\lambda_1|^2 - \lambda_1^{-1})\\
	& \quad + \bar{u}_n p_n \bar{q}_n (\lambda_1 |\lambda_1|^2 - \bar{\lambda}_1^{-1}) +
	(|\lambda_1|^2 |u_n|^2 + |\lambda_1|^{-2})|q_n|^2,
	\end{align*}
	and
	\begin{align*}
	(1 + |u_n|^2) p_{n+1} \bar{q}_{n+1} & = - \lambda_1 u_n |p_n|^2 + \lambda_1 \bar{\lambda}_1^{-1} p_n \bar{q}_n
	-u_n^2 \bar{p}_n q_n + \bar{\lambda}_1^{-1} u_n |q_n|^2.
	\end{align*}
By using the previous expressions and the definitions of $a_n$ and $b_n$, we have verified that the first two equations in system \eqref{5.6} are equivalent to each other and yield the transformation formula \eqref{5.1}. The third equation in system \eqref{5.6} transforms to the relation
	$$
	|a_n|^2 + |b_n|^2 - \bar{a}_n a_{n+1} + \bar{u}_n \bar{a}_n b_{n+1} - u_n \bar{b}_n=0,
	$$
	which holds true after straightforward computations. Finally, the fourth equation in system \eqref{5.6} is consistent with the first three equations after substitutions.
	
	Substituting $V(u_n,\lambda)$ and $M_n(\lambda)$ into \eqref{5.5} and collecting different powers with respect to $\lambda$ yields the system of equations
	\begin{equation}
	\label{5.8}
	\left\{ \begin{array}{l}
	\bar{a}_n \hat{u}_n + b_n - u_n = 0, \\
	\hat{u}_{n-1} - a_{n} u_{n-1} - b_{n} = 0, \\
	\dot{a}_n - i a_{n} \left( u_n \bar{u}_{n-1} + \bar{u}_n u_{n-1} - \hat{u}_n \bar{\hat{u}}_{n-1} -  \bar{\hat{u}}_n \hat{u}_{n-1} - 2 b_n \bar{u}_{n-1} + 2 \bar{b}_n \hat{u}_n \right)=0,\\
	\dot{a}_n - i \left[ a_n (\hat{u}_n \bar{\hat{u}}_{n-1} +  \bar{\hat{u}}_n \hat{u}_{n-1} -
	u_n \bar{u}_{n-1} - \bar{u}_n u_{n-1}) + 2 \bar{b}_n \hat{u}_{n-1} - 2 b_n \bar{u}_n \right]=0,\\
	\dot{b}_n - i \left[ (1+|b_n|^2) \hat{u}_n + b_n (u_n \bar{u}_{n-1} + \bar{u}_n u_{n-1})
	- b_n^2 \bar{u}_{n-1} - \bar{a}_n \hat{u}_{n-1} - a_n u_n + u_{n-1}\right] =0.
	\end{array} \right.
	\end{equation}
	It is obvious that the first two equations in system \eqref{5.8} repeat the first two equations in system \eqref{5.6}, while the third and the fourth equations in system \eqref{5.8} are identical to each other. The third and fifth equations in system \eqref{5.8} are further reduced to the form:
	$$
	\dot{a}_n = i \left[\lambda_1^2 \bar{\lambda}_1^{-2} \bar{b}_n (u_n - b_n) + a_n (u_{n-1} \bar{b}_n - b_n \bar{u}_{n-1}) - b_n (\bar{u}_n - \bar{b}_n)\right],
	$$
	and
\begin{align*}
	\dot{b}_n & = i [ \bar{a}_n^{-1} ( 1+ |b_n|^2) (u_n - b_n) + b_n (u_n \bar{u}_{n-1} + \bar{u}_n u_{n-1}) \\
	& \quad 
	-(1 + |a_n|^2)u_{n-1} - b_n^2 \bar{u}_{n-1} - a_n u_n - \bar{a}_n b_n ].
\end{align*}
It follows from $\dot{\varphi}_n = V(u_n,\lambda) \varphi_n$ that
	\begin{align*}
	\frac{d}{d t} (|\lambda_1|^2 |p_n|^2 + |q_n|^2) &= \frac{i}{2} (\lambda_1^2 + \lambda_1^{-2} -
	\bar{\lambda}_1^2 - \bar{\lambda}_1^{-2}) (|\lambda_1|^2 |p_n|^2 - |q_n|^2) \\
	& \quad + i (|\lambda_1|^4 - 1) (\bar{\lambda}_1^{-1} u_n \bar{p}_n q_n - \lambda_1^{-1} \bar{u}_n p_n \bar{q}_n),
	\end{align*}
	\begin{align*}
	\frac{d}{d t} (|p_n|^2 + |\lambda_1|^2 |q_n|^2) &= \frac{i}{2} (\lambda_1^2 + \lambda_1^{-2} -
	\bar{\lambda}_1^2 - \bar{\lambda}_1^{-2}) ( |p_n|^2 - |\lambda_1|^2 |q_n|^2)  \\
	& \quad + i (|\lambda_1|^4 - 1) (\lambda_1^{-1} u_{n-1} \bar{p}_n q_n - \bar{\lambda}_1^{-1} \bar{u}_{n-1} p_n \bar{q}_n), 
	\end{align*}
	and
	\begin{align*}
	\frac{d}{d t} (p_n \bar{q}_n) &= i (\bar{\lambda}_1 u_{n-1} - \bar{\lambda}_1^{-1} u_n)|p_n|^2
	+ i (\lambda_1 u_n - \lambda_1^{-1} u_{n-1})|q_n|^2  \\
	& \quad + \frac{i}{2} (\lambda_1^2 + \lambda_1^{-2} +
	\bar{\lambda}_1^2 +  \bar{\lambda}_1^{-2} + 2 (u_n \bar{u}_{n-1} + \bar{u}_n u_{n-1})) p_n \bar{q}_n).
	\end{align*}
	We have verified with the help of the Wolfram's Mathematica that the two equations for $\dot{a}_n$ and $\dot{b}_n$ are satisfied by using the previous expressions and the definitions of $a_n$ and $b_n$.
\end{proof}

The following two remarks report applications of the $1$-fold DT with the periodic eigenfunctions of the dnoidal and cnoidal waves. As expected from similar applications in \cite{CPnls,CPW}, the $1$-fold DT recovers the same dnoidal and cnoidal waves
after translations in space and complex phase.

\begin{remark}
	\label{remark-dn}
	Let $\{ u_n(t) \}_{n \in \mathbb{Z}}$ be the dnoidal wave in the form \eqref{2.1} and \eqref{2.3}, and $\{ (p_n(t),q_n(t))^\mathrm{T} \}_{n \in \mathbb{Z}}$ be the eigenfunction of the linear system (\ref{lax-intro}) associated with $\lambda = \lambda_1$ in the form (\ref{3.3}) and (\ref{eigenfunction}). 
	Since $U_n$ and $\lambda_1$ are real, the 1-fold DT \eqref{5.1} yields
	\begin{align*}
	\hat{u}_n &= \left[- \frac{|P_n|^2 + \lambda_1^2 |Q_n|^2}{\lambda_1^2 |P_n|^2 + |Q_n|^2} U_n - \frac{ (\lambda_1^3 - \lambda_1^{-1}) P_n \bar{Q}_n}{\lambda_1^2 |P_n|^2 + |Q_n|^2}\right] e^{2i \omega t}  \\
&= - \frac{\sigma_1 \sqrt{F_1}}{U_n} e^{2i \omega t},
	\end{align*}
where we have used (\ref{eigenfunction}) after multiplying the numerator and the denominator by $\bar{P}_n Q_n$. By using (\ref{2.3}) and (\ref{3-1-0}), we obtain
	\begin{align*}
\hat{u}_n &= - \sigma_1 \frac{ {\rm sn}(\alpha; k)}{{\rm cn}(\alpha; k)} \frac{ \sqrt {1-k^2} }{{\rm dn}(\alpha n;k)} e^{2i \omega t} \\
&= - \sigma_1 \frac{ {\rm sn}(\alpha; k)}{{\rm cn}(\alpha; k)} {\rm dn}(\alpha n + K(k);k) e^{2i \omega t} \\
&= - \sigma_1 u_n(\alpha n + K(k)).
\end{align*}	
The new solution $\hat{u}_n$ is a half-period translation of the dnoidal wave $u_n$ with sign flip for $\sigma_1 = +1$ in case of eigenvalue \eqref{3-1-2}. There is no sign flip for $\sigma_1 = -1$ in case of eigenvalue (\ref{3-1-1}).
\end{remark}

\begin{remark}
	\label{remark-cn}
	Let $\{ u_n(t) \}_{n \in \mathbb{Z}}$ be the cnoidal wave
in the form \eqref{2.1} and \eqref{2.6}, and 
$\{ (p_n(t),q_n(t))^\mathrm{T} \}_{n \in \mathbb{Z}}$ be the eigenfunction of the linear system (\ref{lax-intro}) associated with $\lambda = \lambda_1$ in the form (\ref{3.3}) and  (\ref{eigenfunction}). The 1-fold DT \eqref{5.1} yields:
	\begin{align*}
	\hat{u}_n & = \left[- \frac{\lambda_1(|P_n|^2 + |\lambda_1|^2 |Q_n|^2)}{\bar{\lambda}_1(|\lambda_1|^2 |P_n|^2 + |Q_n|^2)} U_n + \frac{\lambda_1 (1 - |\lambda_1|^4) P_n \bar{Q}_n}{\bar{\lambda}_1^2 (|\lambda_1|^2 |P_n|^2 + |Q_n|^2)}\right] e^{2i \omega t}  \\
	& =  \frac{\lambda_1^2 \left[- F_1 (1- |\lambda_1|^4) - 2\bar{\lambda}_1^2 \sqrt{F_1} U_n^2 + \sqrt{F_1} (|\lambda_1|^4 +1) U_n U_{n-1} \right]}
	{\bar{\lambda}_1^2 \left[\sqrt{F_1} (|\lambda_1|^4 +1) U_n - 2 \lambda_1^2 \sqrt{F_1} U_{n-1} +
		(1- |\lambda_1|^4) U_n^2 U_{n-1} \right]} e^{2i \omega t}  \\
	& = \frac{(\bar{\lambda}_1^{-2} + \lambda_1^2)U_n U_{n-1} - 2U_n^2 - \sqrt{F_1} (\bar{\lambda}_1^{-2} - \lambda_1^2) }
	{(\bar{\lambda}_1^2 + \lambda_1^{-2}) U_n - 2 U_{n-1} + (\sqrt{F_1})^{-1} (\lambda_1^{-2} - \bar{\lambda}_1^2)U_n^2 U_{n-1}} e^{2i \omega t},
	\end{align*}
where we have used (\ref{eigenfunction}) after multiplying the numerator and the denominator by $\bar{P}_n Q_n$. Inserting (\ref{2.6}), (\ref{3-2-2-0}), and (\ref{3-2-2}) into this formula yields
	\begin{equation*}
	\hat{u}_n  =  \frac {k \sqrt{1-k^2} {\rm sn}(\alpha;k) ({\rm cn}(\alpha;k) + i \sqrt{1-k^2} {\rm sn}(\alpha;k))} {{\rm dn}(\alpha;k) ({\rm cn}(\alpha;k) - i \sqrt{1-k^2} {\rm sn}(\alpha;k))} \Upsilon_n e^{2i \omega t},
	\end{equation*}
	where
	\begin{align*}
	\Upsilon_n = & \frac {{\rm cn}(\alpha n;k) {\rm cn}(\alpha n - \alpha;k) - {\rm cn}(\alpha;k) {\rm cn}^2 (\alpha n;k) - i \sqrt{1-k^2} {\rm sn}(\alpha;k) {\rm sn}^2(\alpha n;k)} {\sqrt{1-k^2} ({\rm cn}(\alpha n;k) - {\rm cn}(\alpha;k) {\rm cn}(\alpha n - \alpha;k)) - i {\rm sn}(\alpha;k) {\rm cn}(\alpha n - \alpha;k)
		{\rm dn}^2(\alpha n;k)}  \\
	= & \frac{i ({\rm cn}(\alpha;k) - i \sqrt{1-k^2} {\rm sn}(\alpha;k)) {\rm sn}(\alpha n;k)}
	{{\rm dn}(\alpha;k) {\rm dn}(\alpha n;k)}.
	\end{align*}
After simplification, we obtain
	\begin{align*}
	\hat{u}_n &=  \frac {i k \sqrt{1-k^2} {\rm sn}(\alpha;k) ({\rm cn}(\alpha;k) + i \sqrt{1-k^2} {\rm sn}(\alpha;k)) {\rm sn}(\alpha n;k)} {{\rm dn}^2(\alpha;k) {\rm dn}(\alpha n;k)} e^{2i \omega t} \\
	&= - \frac{i k {\rm sn}(\alpha;k) ({\rm cn}(\alpha;k) + i \sqrt{1-k^2} {\rm sn}(\alpha;k))}{ {\rm dn}^2(\alpha;k)} {\rm cn}(\alpha n + K(k);k) e^{2i \omega t} \\
	&= - \frac{i ({\rm cn}(\alpha;k) + i \sqrt{1-k^2} {\rm sn}(\alpha;k))}{ {\rm dn}(\alpha;k) } u_n(\alpha n + K(k)).
	\end{align*}
Since the amplitude factor has the unit modulus, the new solution is a quarter-period translation of the cnoidal wave (\ref{2.6}) with a suitable phase factor.
\end{remark}

Remarks \ref{remark-dn} and \ref{remark-cn} are used to generate the left panels on Figures \ref{fig-3} and \ref{fig-4}. For the right panels, 
we use the 1-fold DT with the nonperiodic solutions computed in 
the form (\ref{eigenfunction-second}) with $\theta_n(t) = \Theta_n + it$, 
where $\Theta_n$ is computed numerically from (\ref{theta-1-1}) and (\ref{theta-2-7}) for the dnoidal and cnoidal waves respectively. 
The integration constant in $\Theta_n$ is chosen 
from the condition $\Theta_0 = 0$ so that  $\theta_0(0) = 0$.

In order to compute the magnification factors (\ref{5.19}) and (\ref{5.23}) for the rogue waves on the dnoidal and cnoidal background respectively, 
we substitute the second solution in the form (\ref{eigenfunction-second}) with $\theta = 0$ (attained at $n = 0$ and $t = 0$)
into the $1$-fold DT in the form \eqref{5.1}. This yields
$$
\hat{u}_n(t) = \hat{U}_n(t) e^{2 i \omega t}
$$
with
\begin{align}
\hat{U}_n & = - \frac{\lambda_1(|\lambda_1|^2 |P_n|^2 +  |Q_n|^2)}{\bar{\lambda}_1( |P_n|^2 + |\lambda_1|^2 |Q_n|^2)} U_n - \frac{\lambda_1 (1 - |\lambda_1|^4) P_n \bar{Q}_n}{\bar{\lambda}_1^2 (|P_n|^2 + |\lambda_1|^2 |Q_n|^2)}.
\label{5.20}
\end{align}
Multiplying the numerator and the denominator in (\ref{5.20}) by $\bar{P}_n Q_n$ and using \eqref{eigenfunction} yields the explicit relation between the new solution $\hat{U}_n$ and the standing wave $U_n$ for $n = 0$ and $t = 0$:
\begin{align*}
\hat{U}_n &= - \frac{\lambda_1}{\bar{\lambda}_1}
\frac{|\lambda_1|^2 (\bar{\lambda}_1 U_n -
	\bar{\lambda}_1^{-1} U_{n-1})(\sigma_1 \sqrt{F_1} - U_n U_{n-1}) + (\lambda_1 U_{n-1}- \lambda_1^{-1} U_n)(\sigma_1 \overline{\sqrt{F_1}} - U_n U_{n-1})}{(\bar{\lambda}_1 U_n -
	\bar{\lambda}_1^{-1} U_{n-1})(\sigma_1 \sqrt{F_1} - U_n U_{n-1}) + |\lambda_1|^2 (\lambda_1 U_{n-1}- \lambda_1^{-1} U_n)(\sigma_1 \overline{\sqrt{F_1}} - U_n U_{n-1})} U_n \\
& \quad + \frac{\lambda_1}{\bar{\lambda}_1^2}
\frac{(|\lambda_1|^4 -1) (\sigma_1 \sqrt{F_1} - U_n U_{n-1}) (\sigma_1 \overline{\sqrt{F_1}} - U_n U_{n-1})}{(\bar{\lambda}_1 U_n -
\bar{\lambda}_1^{-1} U_{n-1})(\sigma_1 \sqrt{F_1} - U_n U_{n-1}) + |\lambda_1|^2 (\lambda_1 U_{n-1}- \lambda_1^{-1} U_n)(\sigma_1 \overline{\sqrt{F_1}} - U_n U_{n-1})}.
\end{align*}
where we have used that $U_n$ is real-valued for both dnoidal and cnoidal waves.

In the case of the dnoidal wave (\ref{2.3}), we use the fact that $\lambda_1 \in \mathbb{R}$ and $F_1 > 0$ to simplify the expression for $\hat{U}_n$ to the form
\begin{equation*}
\hat{U}_n =- \frac{ U_n^2 + U_n U_{n-1} - \sigma_1 \sqrt{F_1}}{U_{n-1}}
\end{equation*}
which together with \eqref{2.3} and \eqref{3-1-0}, and \eqref{3-1-1} yields
\begin{equation*}
|\hat{U}_0(0)| = \frac{{\rm sn}(\alpha;k) (1 + {\rm dn}(\alpha;k) - \sigma_1 \sqrt
	{1-k^2} )}{{\rm cn}(\alpha;k) {\rm dn}(\alpha;k)}.
\end{equation*}
Dividing this formula to $A = {\rm sn}(\alpha;k)/{\rm cn}(\alpha;k)$ yields the magnification factor for the dnoidal wave in the form (\ref{5.19}).

In the case of the cnoidal wave (\ref{2.6}), we use the fact that $\lambda_1 \in \mathbb{C}\backslash \mathbb{R}$, $F_1 < 0$, and $\sigma_1 = 1$ to simplify the expression for $\hat{U}_n$ to the form
\begin{align*}
\hat{U}_0(0) = \frac{k {\rm sn}(\alpha;k) ({\rm dn}(\alpha;k) + 1) 
	({\rm cn}(\alpha;k) + i \sqrt{1-k^2} {\rm sn}(\alpha;k))}{{\rm dn} (\alpha;k) [k^2 {\rm cn}^2(\alpha;k) + (1-k^2)]},
\end{align*}
where we have used \eqref{2.6}, \eqref{3-2-2-0}, and \eqref{3-2-2}. 
This yields
\begin{align*}
|\hat{U}_0(0)|
= \frac{k {\rm sn}(\alpha;k) ({\rm dn}(\alpha;k) +1)}{{\rm dn}^2 (\alpha;k)}.
\end{align*}
Dividing by  $A = k {\rm sn}(\alpha;k)/{\rm dn}(\alpha;k)$ yields the magnification factor for the cnoidal wave in the form (\ref{5.23}).

\vspace{0.25cm}

{\bf Acknowledgements.} This work was supported in part by the National
Natural Science Foundation of China (No. 11971103) and the Project ``333'' of Jiangsu Province.

\end{document}